\documentclass{article}
\usepackage[margin=1in]{geometry}
\usepackage[utf8]{inputenc}

\usepackage{balance}
\usepackage{graphicx}
\usepackage{tikz}
\usepackage{subcaption}
\usepackage{caption}
\usepackage[dvipsnames]{xcolor}
\usetikzlibrary{positioning, shapes.geometric}

\usepackage{amssymb,amsfonts,mathtools, xcolor, comment, amsthm,url,stmaryrd}
\usepackage{cleveref}
\usepackage{booktabs}
\usepackage[ruled,vlined]{algorithm2e}

\newtheorem{lemma}{Lemma}

\newtheorem{definition}{Definition}
\newtheorem{remark}{Remark}

\newtheorem{proposition}{Proposition}
\newtheorem{theorem}{Theorem}

\crefname{figure}{Fig.}{Fig.}
\crefname{example}{Example}{Example}
\crefname{proposition}{Proposition}{Proposition}
\crefname{corollary}{Corollary}{Corollary}
\crefname{theorem}{Theorem}{Theorem}
\crefname{lemma}{Lemma}{Lemma}
\crefname{assumption}{Assumption}{Assumption}
\crefname{appendix}{Appendix}{Appendix}
\crefname{definition}{Definition}{Definition}
\crefname{remark}{Remark}{Remark}
\crefname{table}{Table}{Table}
\crefname{example}{Example}{Example}
\crefname{algocf}{Algorithm}{Algorithms}

\newcommand{\bu}{\boldsymbol{u}}
\newcommand{\bv}{\boldsymbol{v}}
\newcommand{\ba}{\boldsymbol{a}}
\newcommand{\bw}{\boldsymbol{w}}
\newcommand{\bx}{\boldsymbol{x}}

\newcommand{\bq}{\boldsymbol{q}}
\newcommand{\bz}{\boldsymbol{z}}
\newcommand{\by}{\boldsymbol{y}}
\newcommand{\be}{\boldsymbol{e}}
\newcommand{\edit}[1]{{\color{black} #1}}

\title{Learning to Recommend in Unknown Games}
\author{Arwa Alanqary \thanks{University of California, Berkeley}
\and Zakaria Baba \thanks{California Institute of Technology}
\and Manxi Wu \footnotemark[1]
\and Alexandre M. Bayen \footnotemark[1]
}
\date{}

\begin{document}

\maketitle

\begin{abstract}
We study preference learning through recommendations in multi-agent game settings, where a moderator repeatedly interacts with agents whose utility functions are unknown. In each round, the moderator issues action recommendations and observes whether agents follow or deviate from them. We consider two canonical behavioral feedback models—best response and quantal response—and study how the information revealed by each model affects the learnability of agents’ utilities. We show that under quantal-response feedback the game is learnable, up to a \edit{positive} affine equivalence class, with logarithmic sample complexity in the desired precision, whereas best-response feedback can only identify a larger set of agents' utilities. We give a complete geometric characterization of this set. Moreover, we introduce a regret notion based on agents’ incentives to deviate from recommendations and design an online algorithm with low regret under both feedback models, with bounds scaling linearly in the game dimension and logarithmically in time. Our results lay a theoretical foundation for AI recommendation systems in strategic multi-agent environments, where recommendation compliances are shaped by strategic interaction.
\end{abstract}

\section{Introduction}
\label{sec:intro}
Modern digital platforms driven by AI algorithms increasingly operate as intermediaries among strategic users, providing recommendations, answering queries, or allocating shared resources. Examples include route guidance systems in traffic networks \cite{acemoglu2018informational,gollapudi2023online,vu2021fast}, bidding and pricing assistants in online marketplaces \cite{liu2021optimal,calvano2020artificial}, and ranking mechanisms in online marketplaces where sellers compete for demand \cite{yao2024user}. In these settings, the platform does not control users’ actions and cannot observe users' numerical utilities. Instead, the platform issues recommendations (suggested actions such as prices, bids, and route recommendations), and observes whether users follow the suggestion or deviate. These responses reflect users’ incentives, which are strategically coupled through the underlying game. In such settings, we study the following question:
\begin{center}
    \emph{How can a platform issue recommendations that are compliant for multiple strategic agents when their utilities are unknown and only action feedback is observed?}
\end{center}

This question is central to the design of AI systems that operate in multi-agent strategic environments, where the effectiveness of recommendations depends not only on how well they align with an individual user’s preferences, but also on how users’ incentives interact with one another. Existing work on AI–human alignment and preference learning largely focuses on the interaction between an algorithm and a single agent \cite{zadimoghaddam2012efficiently,beigman2006learning,ng2000algorithms,hadfield2016cooperative,yao2022learning}. In contrast, in multi-agent setting, a recommendation that is aligned to a user's preference in isolation may be ignored once users anticipate the behavior of others who receive correlated recommendations \cite{aumann1974subjectivity}. As a result, whether an agent follows a recommendation is inherently a strategic decision, shaped by beliefs about others’ actions and the underlying game structure. This strategic dependence fundamentally alters both what can be learned from observed behavior and how recommendations should be designed to induce compliance in multi-agent strategic environment.


\paragraph{Our contributions}
We formalize this problem from the perspective of a moderator interacting repeatedly with a fixed set of strategic agents playing an unknown normal-form game for $T$ rounds. In each round, the moderator commits to a recommendation mechanism that privately suggests actions to agents. Agents then decide whether to follow or deviate from these recommendations, based on their own utilities and beliefs induced by the recommendation. We study two canonical models of agent behavior. In best-response (BR) model, agents choose an action that maximizes their expected utility given the recommendation. In quantal-response (QR) model, agents behave in a boundedly-rational manner, selecting actions probabilistically based on their incentive to deviate. The moderator observes the realized actions, not the utilities.

A natural benchmark for a ``compliant" recommendations in this environment is correlated equilibrium (CE). A recommendation mechanism is a CE if no agent has an incentive to deviate from the recommended action, given the induced beliefs. Accordingly, we evaluate the moderator’s performance by the incentives to deviate induced by its recommendations: recommendations are good if agents have little to gain from deviation. Our focus is to answer the following questions: 

\begin{itemize}
\item[-] \textbf{Learnability}: 
\edit{Can the moderator recover the unknown agents' utilities from repeated recommendation and action feedback? If not, which equivalence classes of agents’ utility functions are identifiable?}

\item[-] \textbf{Regret-minimization}: Can the moderator learn a recommendation mechanism that achieves low incentive-to-deviate regret over time?
\end{itemize}


A game is said to be learnable under a given feedback model if the model allows the moderator to identify agents’ utilities up to an agent-wise positive affine transformation that preserves the equilibrium solution set. We show that the game is learnable under the QR feedback model. In contrast, under the BR feedback model, a strictly larger class of utility transformations remains indistinguishable; we fully characterize this class. 
\edit{This characterization may be of independent interest because it characterizes such transformations in a larger set of inverse problems.}


\begin{theorem}[Informal learnability results]
The game utility functions are learnable under the QR model, but not under the BR model. Moreover, we characterize the set of all transformations of a game that keeps it indistinguishable under the BR model using polyhedral duality. 
\end{theorem}


\edit{Let $n$ be the number of agents in the game, $m$ be the maximum number of actions for each agent, and $M$ be the number of action profiles. We show that under the QR feedback, the moderator can learn the game with logarithmic complexity in the desired precision and (near-)linear dependence on the normal-form representation size $nM$, up to an additional multiplicative factor $m$.}

 
\begin{theorem}[Informal learning complexity result]
\edit{Under the QR feedback, we construct an algorithm that learns the utilities of the underlying game up to agent-wise positive affine transformation with ~$\epsilon$-precision using~$O(mnM\log(1/\epsilon))$ recommendations.}
\end{theorem}

Furthermore, under both BR and QR feedback models, we design an online algorithm that produces low-regret recommendations, where regret is measured by the agents’ accumulated incentives to deviate from the recommended actions. Our approach reduces the problem of low-regret recommendation to a geometric cutting-plane problem inspired by contextual search and inverse optimization literature \cite{gollapudi2021contextual,liu2021optimal,besbes2025contextual} via a novel oracle construction. 
\edit{The regret of this recommendation algorithm scales linearly in the normal-form representation size $nM$ and logarithmically in the number of rounds $T$.}
\begin{theorem}[Informal regret bound] We construct an algorithm for the moderator to give recommendations with \edit{$O(nM\log(T))$} regret under the BR and QR feedback models. \end{theorem}

\subsection{Related work}
\paragraph{Inverse game theory}
Inverse game theory (IGT) studies the problem of inferring or rationalizing players’ utility functions from observed behavior in a game. Given an observed outcome, the objective is to identify utility functions that could have generated that behavior. Central to most inverse game-theoretic approaches is the assumption that the observed outcome corresponds to an (approximate) equilibrium of the underlying game \cite{kuleshov2015inverse,beigman2006learning,waugh2013computational}. This assumption is often restrictive and requires information about the game that may not be available to the moderator.

Our approach departs from the classical IGT setting by not assuming knowledge of any equilibrium outcome of the underlying game. Instead, we adopt an active learning perspective in which a moderator can probe the game to elicit informative responses. By explicitly leveraging off-equilibrium behavior, the moderator can infer substantially more about agents’ incentives than can be identified from equilibrium observations alone. Indeed, a well-known limitation of inverse game theory is underdetermination: many distinct games can rationalize the same observed equilibrium, including degenerate cases such as games with constant payoffs across outcomes and players. Observations of off-equilibrium outcomes help prune this hypothesis space and rule out many such trivial explanations. We show that, under a reasonable feedback model in our setting, it is possible to identify the true payoff functions of the game up to natural invariance classes (cf. \cref{sec:learnability}).

Closely related to our setting is the recent work of \cite{zhang2025learning}, which also considers a moderator seeking to learn unknown utilities in a game under different behavioral models for the agents. In their framework, however, the moderator influences the game through a combination of signaling and exogenous payments that additively modify agents’ utilities. In particular, the moderator can alter the game by selecting payment functions that directly change each agent’s rewards, and these payments serve as the main mechanism for eliciting preferences. We address a different setting in which the moderator has less influence. The only available lever is the ability to issue non-binding recommendations. As a result, our approach is applicable to a broader class of environments in which monetary incentives or direct utility modifications are infeasible. 

\paragraph{Learning optimal Stackelberg strategies}
Another related line of work studies the problem of learning optimal Stackelberg strategies when follower payoffs are unknown. In these settings, a leader repeatedly interacts with one or more followers to learn their utilities and compute an optimal commitment strategy. The leader is an explicit player with its own utility function and an action space that allows it to influence the equilibrium of the game \cite{zhao2023online,wu2022inverse,marecki2012playing,peng2019learning,personnat2025learning}.

Our work differs primarily in the power and role of the moderator. In our setting, the moderator does not take actions in the game but instead issues recommendations. Moreover, most prior work in this area focuses on a single follower, with more recent contributions \cite{personnat2025learning} extending to multiple followers under the assumption of no strategic externalities among them. In contrast, our setting accommodates multiple strategically coupled agents without granting the moderator direct control over outcomes or incentives.

\paragraph{Contextual search and inverse optimization}
Inverse optimization is a general and powerful paradigm for learning latent preferences and constraints that rationalize the observed behavior of a decision maker. Recent work on online and contextual inverse optimization extends this paradigm to online environments, where a learner repeatedly observes a changing context and the decisions made by agents in response to that context \cite{gollapudi2021contextual,dong2018generalized,besbes2025contextual}. While inverse optimization typically focuses on a single decision maker and treats the context as exogenous, it provides foundational tools for inferring preferences from actions, which we build upon in our multi-agent setting (cf. \cref{sec:low_reg_alg}). A key distinction in our work is that the learner (interpreted as a platform or assistant) endogenously influences the context by issuing recommendations or information to agents, rather than passively observing decisions under externally given contexts. 

Another related line of work studies the data requirements needed to recover uncertain parameters or to identify an optimal decision in contextual and data driven optimization problems \cite{besbes2023big, bennouna2026data,hu2022fast}. In particular, \cite{bennouna2026data} provide a geometric characterization of when observed data are sufficient to recover optimal decisions in linear optimization. This perspective is connected to our learnability question in Section 3, while our focus is to study which feedback models are sufficient to recover the unknown utilities (up to natural equivalences).


\section{Model and preliminaries}
\label{sec:model}
We consider a finite normal form game $G(\{A_i\}_{i=1}^n, \{u_i\}_{i=1}^n)$ with $n$ agents. Each agent has a set of $m_i$ actions
$A_i = \{a_i^{1}, \dots a_i^{m_i}\}$ and utility function $u_i: A\to \mathbb{R}$, where $A \coloneqq \prod_{i=1}^n A_i$ is the joint-action set. \edit{Let $m \coloneqq \max_{i\in[n]} m_i$ and $M\coloneqq |A| = \prod_{i}^n m_i$ be the cardinality of the joint action set} . An action profile is denoted by $\ba \in A$ and $\ba_{-i} \in A_{-i}$ denotes the action profile of all agents except $i$. We denote by $d_i$ the cardinality of the opponents' joint action space $A_{-i}$, given by $d_{i} \coloneqq |A_{-i}| = \prod_{j \neq i} m_j$.
For a fixed agent $i \in [n]$ and action $a_i \in A_i$ we define the utility vector  
\[
\bu_i(a_i, \cdot) \in \mathbb{R}^{d_i}, \quad  \bu_i(a_i, \cdot)= \left[\bu_i(a_i, \ba_{-i})\right]_{\ba_{-i} \in A_{-i}} 
\]
which lists agent $i$’s utilities across all opponent action profiles. For a pair of actions  $a_i, a_i' \in A_i$ define the utility difference vector
\[
\bw_i(a_i, a_i') \in \mathbb{R}^{d_i}, \quad  \bw_i(a_i, a_i') = \bu_i(a'_i, \cdot) - \bu_i(a_i, \cdot).
\]

Let $\Delta(A) \subset \mathbb{{R}}_{+}^{M}$ be the set of joint distributions over action profiles. For a given distribution $\bx \in \Delta(A)$, let $\bx(\ba)$ be the probability assigned to action profile $\ba$. For a fixed action $a_i$ for agent $i$ we denote the vector 
\[
\bx(a_i, \cdot) \in \mathbb{R}^{d_i}, \quad  \bx(a_i, \cdot)= \left[\bx(a_i, \ba_{-i})\right]_{\ba_{-i} \in A_{-i}}.
\]

\paragraph{Recommendations in unknown games.} We take the perspective of a moderator who repeatedly interacts with $n$ agents playing the game $G(\{A_i\}_{i=1}^n, \{u_i\}_{i=1}^n)$. The moderator knows the game structure (i.e. the number of agents and the action sets $\{A_i\}_{i=1}^n$) but does not know the agents' utility functions $\{u_i\}_{i=1}^n$.
The interaction proceeds over $T$ rounds. In each round $t$ the moderator commits to a recommendation mechanism in the form of a probability distributions over action profiles $\bx^{(t)} \in \Delta(A)$. The moderator then samples an action profile $\ba\sim \bx^{(t)}$ and privately recommends to each agent $i$ their corresponding action $a_i$. 
Agents observe the recommendation mechanism $\bx^{(t)}$ and their private action recommendation $a_i$, but not the actions recommended to others. Agents then select, through a \emph{choice model}, which action $a_i^{\star}$ to play. The moderator observes the realized actions $\ba^{\star}$ for all agents, but receives no information about their utilities.

\paragraph{Choice models.} Upon observing a recommended action $a_i$ from a recommendation mechanism $\bx$, agent $i$ ranks actions according to ${u^{\bx}_i(a_i'| a_i)} = \langle \bx(a_i, \cdot), \bu_i(a'_i, \cdot)\rangle$ \footnote{The conditional expected utility of playing $a'_i$ when receiving recommendation $a_i$ is $u^{\bx}_i(a_i'| a_i)/\mathrm{Pr}^{\bx}(a_i)$, where $\mathrm{Pr}^{\bx}(a_i)$ is the marginal probability of  $a_i$ in $\bx$. However, since $\mathrm{Pr}^{\bx}(a_i)$ does not affect the ranking of actions in $A_i$, we drop this from the analysis.}. The recommendation mechanism along with the recommended action forms the agents belief about the joint actions of others and the quality of an action depends on its expected utility with respect to this belief. 
A recommendation mechanism $\bx$ is said to be a \emph{correlated equilibrium} if for every action profile $\ba \sim \bx$, the recommended action $a_i$ for each agent $i$ is an optimal action. 
\begin{definition}\label{def:CE}
A probability distribution $\bx \in \Delta(A)$ is a correlated equilibrium (CE) of the game \\ $G(\{A_i\}_{i=1}^n, \{u_i\}_{i=1}^n)$ if for every agent $i \in [n]$ and every actions $a_i \in A_i$ we have 
\begin{align} 
\varphi_i(a_i, a_i', \bx) \coloneqq \left\langle \bx(a_i, \cdot), \bw_i(a_i, a_i') \right\rangle & \leq 0, & \forall a_i'\in A_i, 
\end{align}  
where $\bw_i(a_i, a_i') = \bu_i(a_i', \cdot) - \bu_i(a_i, \cdot)$.
\end{definition}
Following \cref{def:CE}, a recommendation mechanism $\bx$ is an approximate correlated equilibrium, denoted by $\epsilon$-CE if for every agent $i$, every actions $a_i \in A_i$, and every deviation $a_i'\in A_i$, we have $\varphi_i(a_i, a_i', \bx) \leq \epsilon$. The value of $\varphi_i(a_i, a_i', \bx)$ represents the incentive agent $i$ has to deviate from the recommended action $a_i$ to a different action $a_i'$.

The action the agent ends up choosing reveals information about their ordinal preferences (ranking) over actions. We will discuss two specific choice models and the information they reveal to the moderator. 
First, we consider a best responding agent, who reports their best ranked action given the recommendation.

\begin{definition}[Best-response sets] 
\label{def:best_response}
For each recommendation mechanism $\bx$, agent $i \in [n]$, and recommended action $a_i$, define the best-response set 
\begin{align}
    BR_i(\bx, a_i) =\arg\max_{a_i'\in A} u^{\bx}_i(a_i'|a_i). 
\end{align}
\end{definition}
Under the best-response feedback model, the realized actions the moderator observes in response to a recommendation $\ba \sim \bx$ are drawn uniformly at random from the best-response sets $a_i^{\star} \in BR_{i}(\bx, \ba_i)$.

Next, we consider a boundedly rational agent following a quantal-response model. Given a recommendation mechanism $\bx$ and an action recommendation $a_i$, an agent selects a deviation $a'_i$ with probability proportional to its incentive to deviate \footnote{Unlike the classical quantal response model, which assigns positive probability to all actions based on their expected utilities, this model assigns probability only to deviations that weakly improve upon the recommended action. This asymmetry reflects the informational role of recommendations: it serves as a reference to evaluate alternatives and only beneficial deviations are considered.}
\begin{align}
\label{eq:quantal_prob}
    P_{\bx}(a'_i|a_i) \propto \begin{cases}
        \exp\left(\beta \varphi_i(a_i, a_i', \bx)\right) & \text{if } \varphi_i(a_i, a_i', \bx)\geq 0   \\
        0 & \text{otherwise}
    \end{cases}
\end{align}
where $\beta\geq 0$ is a rationality parameter.

\edit{Since the moderator does not know the agents' exact quantal-response probability model (e.g., the parameter $\beta$), we do not use relative action frequencies. Instead, we restrict attention to the coarser feedback captured by the quantal-response sets which describe the set of actions actions that can occur with positive probability under some quantal-response behavior}.

\begin{definition}[Quantal-response sets] 
\label{def:quantal_response}
For each recommendation mechanism $\bx$, agent $i \in [n]$, and recommended action $a_i$, define the quantal-response set
\begin{align}
    QR_i(\bx, a_i) = \{a_i' \in A_i| \varphi_i(a_i, a_i', \bx)\geq 0\}.
\end{align}
\end{definition}
\paragraph{Regret.} 
The goal of the moderator is to issue good recommendations that ensure agents’ compliance. For a given recommendation mechanism
$\bx$ and action profile $\ba \sim \bx$, and realized action profile $\ba^{\star}$, we define the moderator's regret as the aggregate incentive of agents to deviate from the recommended actions to the chosen actions 
\begin{align}
    r(\ba, \ba^{\star}, \bx) = \sum_{i\in [n]} \varphi_i(a_i, a_i^{\star}, \bx).
\end{align}
Under both the best-response and quantal-response feedback models, this regret is always nonnegative, since agents only deviate to actions with non-negative incentives. In our context, correlated equilibria capture the compliance constraints for all agents in the game. The regret $r(\ba, \ba^{\star}, \bx) $ can be viewed as a per-sample violation of the CE constraints. If $\bx$ is a CE of the game then $r(\ba, \ba^{\star}, \bx)=0$ for all samples $\ba$ and possible deviations $\ba^{\star}$. 

We consider algorithms $\mathcal{A}$ that select a sequence of recommendation mechanisms $\{\bx^{(t)}\}_{t=1}^T$ over $T$ rounds. The performance of the algorithm is measured by its expected cumulative regret
\begin{align}
    \operatorname{Reg}(\mathcal{A}) = \mathbb{E}\left[ \sum_{t=1}^T r(\ba^{(t)}, \ba^{\star (t)}, \bx^{(t)}) \right], 
\end{align}
where $\ba^{(t)}$ is the sampled action profile at time $t$ and $\ba^{\star (t)}$ is the realized action profile at time $t$. The expectation is taken with respect to the randomness of the algorithm, the sampling of action profiles $\ba^{(t)} \sim\bx^{(t)}$, and the stochasticity of agents’ responses. Namely, the expectation is taken over the joint distribution of the trajectory $\left\{ \left(\ba^{(t)}, \ba^{\star (t)}, \bx^{(t)}\right)\right\}_{t=1}^T$. We say an algorithm $\mathcal{A}$ is low-regret if $\operatorname{Reg}(\mathcal{A}) = o(T)$. 


\section{Learnability}
\label{sec:learnability}
The first question we address is whether a given feedback model is, in principle, sufficiently rich for the moderator to learn the underlying game utilities $\{u_i\}_{i=1}^n$. However, exact recovery of utilities is impossible. Indeed, the feedback sets generated by both the best-response and quantal-response models are fully determined by the ordinal comparisons of expected payoffs, $u_i^{\bx}(a_i' | a_i) = \langle \bx(a_i, \cdot), \bu_i(a'_i, \cdot)\rangle$, which are invariant under agent-specific positive affine transformations of the utility vectors $\bu_i(a_i',\cdot)$.
Motivated by this, we introduce the notion of strategically \emph{equivalent games}. 

\begin{definition}[Game equivalence]
\label{def:game_eq}
    Two games $G(\{A\}_{i=1}^n, \{u_i\}_{i=1}^n)$ and $G'(\{A\}_{i=1}^n, \{v_i\}_{i=1}^n)$ are equivalent if for every agent $i \in [n]$ and action $a_i \in A_i$, 
    \begin{align}
    \label{eq:pos_affine}
        \bv_i(a_i, \cdot) = \lambda_i \bu_i(a_i, \cdot) + \boldsymbol{t}_i
    \end{align}
    for some $\lambda_i >0$ and $\boldsymbol{t}_i \in \mathbb{R}^{d_i}$.
\end{definition}
\begin{remark}
    Equivalent games induce identical strategic outcomes under a broad class of agent behavioral models. They generate the same best-response sets, quantal-response supports, and equilibrium sets under standard equilibrium notions such as Nash, correlated, and coarse correlated equilibrium.
\end{remark}

\paragraph{Indistinguishability.}
To formalize the learnability question, we need to reason about the information
available to the moderator. A moderator interacts with the game only through a specific feedback model and may choose any recommendation mechanisms. Two games are indistinguishable to the moderator if they generate identical feedback under all such interactions.

\begin{definition}[Indistinguishability under a feedback model]
\label{def:obs_eq}
Fix a feedback model $\mathcal{F}$. Two games \\
$G(\{A_i\}_{i=1}^n,\{u_i\}_{i=1}^n)$ and
$G'(\{A_i\}_{i=1}^n,\{v_i\}_{i=1}^n)$ are indistinguishable under
$\mathcal{F}$ if for every recommendation mechanism $\bx \in \Delta(A)$, the feedback generated by $\mathcal{F}$ is identical in the two games.
\end{definition}

We instantiate this notion for the two feedback models considered in this work: 
\begin{itemize}
    \item[-] \textbf{Best-response feedback}: Under the best-response model, indistinguishability requires that for every recommendation mechanism
    $\bx \in \Delta(A)$, and every agent $i$ and action $a_i \in A_i$,
    \[
        \arg\max_{a_i' \in A_i} u_i^{\bx}(a_i' \mid a_i)
        =
        \arg\max_{a_i' \in A_i} v_i^{\bx}(a_i' \mid a_i).
    \]
    \item[-] \textbf{Quantal-response feedback}: Under the quantal-response model, indistinguishability requires that for every recommendation mechanism
    $\bx \in \Delta(A)$, and every agent $i$ and action $a_i \in A_i$,
    \[
        \{ a_i' \in A_i \mid
        \langle \bx(a_i,\cdot),
        \bu_i(a_i',\cdot) - \bu_i(a_i,\cdot) \rangle \ge 0 \}
        =
        \{ a_i' \in A_i \mid
        \langle \bx(a_i,\cdot),
        \bv_i(a_i',\cdot) - \bv_i(a_i,\cdot) \rangle \ge 0 \}.
    \]
\end{itemize}

\paragraph{Learnability.}
We are now ready to define learnability of a game. Intuitively, a game is learnable from a feedback model if the feedback rules out all games that are not equivalent to it.

\begin{definition}[Learnability from a feedback model]
\label{def:learnability}
A game $G(\{A_i\}_{i=1}^n,\{u_i\}_{i=1}^n)$ is \emph{learnable} from a feedback model $\mathcal{F}$ if every game that is indistinguishable to it under
$\mathcal{F}$ is equivalent in the sense of \cref{def:game_eq}.
\label{def:learnability_qr}
\label{def:learnability_br}
\end{definition}

In section \cref{sec:learnability_qr}, we establish positive learnability results for a general class of games under the qunatal-response model. In \cref{sec:learnability_br} we show that the same class of games is not learnable
from the best-response feedback. We provide a geometric characterization of the set of all games that are indistinguishable under the best-response feedback.

\subsection{Learnability from the quantal-response feedback}
\label{sec:learnability_qr}
In this section, we discuss the learnability of a class of games from the quantal-response feedback model. Specifically, we consider games with no \emph{weakly dominated} actions. This condition ensures that every pair of actions exhibits payoff tradeoffs across opponents’ actions, which will be essential for learnability from sign-based feedback. \edit{It also ensures that every action $a_i'$ can be realized as a best-response to some recommendation mechanism $\bx$ and action recommendation $a_i$, for every agent $i$.}
\begin{definition}[Weakly dominated actions]
    An action $a_i \in A_i$ for agent $i$ is {weakly dominated} if there exists another action $a'_i \neq a_i$ such that $u_i(a'_i, a_{-i}) \geq u_i(a_i, a_{-i})$ for all $a_{-i} \in A_{-i}$ and $u_i(a'_i, a_{-i}) > u_i(a_i, a_{-i})$ for some $a_{-i} \in A_{-i}$.
\end{definition}

In \cref{thm:learnability_qr} we establish that any generic game with no weakly dominated strategy is learnable from the quantal-response feedback in the sense of \cref{def:learnability_qr}.  

\begin{theorem}
\label{thm:learnability_qr}
    A {generic game} with {no weakly dominated actions} is learnable from the qunatal-response feedback. 
\end{theorem}

The proof of \cref{thm:learnability_qr} proceeds in two steps. First, we show that quantal-response feedback identifies each pairwise utility difference vector $\bw_i(a_i, a_i')$ up to a positive scalar. Second, we show that these scalars must be consistent across all actions of an agent, yielding identification up to the positive affine transformation in \cref{def:game_eq}. A key technical ingredient is the following geometric lemma, which formalizes when sign information over all nonnegative directions identifies a vector up to scale. We defer the proof of \cref{lem:same_halfspace} to \cref{proof:same_halfspace}.


\begin{lemma}
\label{lem:same_halfspace}
Let $\bw, \bw'\in\mathbb{R}^d$ be two vectors each having at least one strictly positive and one strictly negative component. If for every $\bx \in \mathbb{R}^d_{+}$ we have $\operatorname{sign}\langle \bx, \bw\rangle = \operatorname{sign}\langle \bx,\bw'\rangle$, then there exists $\lambda > 0$ such that $\bw'=\lambda \bw$.
\end{lemma}

\begin{proof}[Proof of \cref{thm:learnability_qr}]
Fix a game $G(\{A_i\}_{i=1}^n,\{u_i\}_{i=1}^n)$ and let $G'(\{A_i\}_{i=1}^n,\{v_i\}_{i=1}^n)$ be any game indistinguishable from $G$ under
the quantal-response feedback. 
By \cref{def:quantal_response}, for every recommendation mechanism $\bx \in \Delta(A)$, the feedback reveals the sign of
\[
\left\langle \bx(a_i,\cdot), \bw_i(a_i,a_i') \right\rangle,
\qquad
\bw_i(a_i,a_i') := \bu_i(a_i',\cdot) - \bu_i(a_i,\cdot),
\]
for every agent $i$ and action pair $a_i,a_i' \in A_i$. Let $\bw_i'(a_i,a_i') := \bv_i(a_i',\cdot) - \bv_i(a_i,\cdot)$.

\paragraph{Step 1: Identification of utility differences up to scale.} Since signs are invariant to scaling of $\bx$, and the moderator may query any
$\bx \in \Delta(A)$, the feedback we get from the quantal-response is equivalent to
\[
\operatorname{sign}\langle \by, \bw_i(a_i,a_i') \rangle
=
\operatorname{sign}\langle \by, \bw_i'(a_i,a_i') \rangle
\quad
\forall \by \in \mathbb{R}^{d_i}_+.
\]
Because the game has no weakly dominated actions, each $\bw_i(a_i,a_i')$ has at least one strictly positive and one strictly negative
component. By \cref{lem:same_halfspace}, there exists $\lambda_{i,a_i,a_i'} > 0$
such that
\[
\bw_i'(a_i,a_i') = \lambda_{i,a_i,a_i'} \bw_i(a_i,a_i').
\]

\paragraph{Step 2: Uniform scaling across actions.}
We now show that $\lambda_{i,a_i,a_i'}$ must be uniform across all action pairs for a given agent. Fix agent $i$ and three actions $a,b,c \in A_i$ (the case $|A_i|=2$ is immediate). The utility differences of game $G'$ must satisfy the triangular identity 
\[
\bw'_i(a,c) = \bw'_i(a,b) + \bw'_i(b,c),
\]
hence,
\[
\lambda_{i,a,c}\bw_i(a,c)
=
\lambda_{i,a,b}\bw_i(a,b) + \lambda_{i,b,c}\bw_i(b,c).
\]
Rewriting in terms of utilities yields
\[
(\lambda_{i,a,c} - \lambda_{i,a,b}) \big(\bu_i(a, \cdot) - \bu_i(c, \cdot) \big)   
+ 
(\lambda_{i,a,b} - \lambda_{i,b,c}) \big(\bu_i(b, \cdot) - \bu_i(c, \cdot) \big) 
= 
0.
\]
By genericity, the vectors $\bu_i(a, \cdot), \bu_i(b, \cdot)$ and $\bu_i(c, \cdot)$ are not collinear, thus the vectors $\big(\bu_i(a, \cdot) - \bu_i(c, \cdot) \big)$ and $\big(\bu_i(b, \cdot) - \bu_i(c, \cdot) \big)$ are linearly independent implying that the coefficients must vanish, hence, $\lambda_{i,a,c} =  \lambda_{i,a,b} = \lambda_{i,b,c}$. Since the triple $(a,b,c)$ is arbitrary, we conclude that there exists an agent specific scalar $\lambda_i > 0$ such that
\[
\bv_i(a_i',\cdot)-\bv_i(a_i,\cdot)
=
\lambda_i\big(\bu_i(a_i',\cdot)-\bu_i(a_i,\cdot)\big)
\quad
\forall a_i,a_i' \in A_i.
\]

\paragraph{Step 3: Reconstruction up to equivalence.}
Fix an arbitrary reference action $\bar a_i \in A_i$ and define
$\boldsymbol{t}_i := \bv_i(\bar a_i,\cdot)-\lambda_i\bu_i(\bar a_i,\cdot)$.
Then for all $a_i \in A_i$,
\[
\bv_i(a_i,\cdot) = \lambda_i\bu_i(a_i,\cdot) + \boldsymbol{t}_i.
\]
Thus $G'$ is equivalent to $G$ in the sense of \cref{def:game_eq}, so $G$ is
learnable from quantal-response feedback.

\end{proof}

\subsection{Learnability from the best-response feedback}
\label{sec:learnability_br}

In this section, we also consider the class of games with no weakly dominated actions and we prove negative results on learnability from the best-response feedback. We prove this through a concrete counterexample showing two games that are not equivalent yet remain indistinguishable under best-response feedback.

\begin{theorem}
\label{thm:learnability_br}
    In the class of generic games with {no weakly dominated actions}, there exists games that are not learnable from the best-response feedback. 
\end{theorem}

\begin{proof} A game is \textit{not} learnable form the best-response feedback if there exists another,
non-equivalent game that generates identical best-response feedback for all recommendation mechanisms. We prove the theorem by explicit construction of two such games.

Consider two-player games $G(\{A_i\}_{i=1}^2, \{u_i\}_{i=1}^2)$ and $G'(\{A_i\}_{i=1}^2, \{v_i\}_{i=1}^2)$ with action sets \\
$A_1 = \{a^1, a^2, a^3, a^4\}$ and $A_2 = \{b^1, b^2\}$. The utilities of player~1 in the two games are given by
\begin{align*}
\ u_1= \begin{array}{c|cccc}
 & a^1 & a^2 & a^3 & a^4 \\
\hline
b^1 & 0 & 3 & 5 & 8 \\
b^2 & 8 & 6.5 & 4.5 & 0
\end{array}, 
\quad v_1= \begin{array}{c|cccc}
 & a^1 & a^2 & a^3 & a^4 \\
\hline
b^1 & 0 & 2 & 6 & 8 \\
b^2 & 8 & 7 & 3 & 0
\end{array}. 
\end{align*}
The utilities of player~2 are identical in both games.

First, both games have no weakly dominated actions. Indeed, for every distinct
pair $a_1 \neq a_1' \in A_1$, the utility difference vectors
$\bw_1(a_1,a_1') = \bu_1(a_1',\cdot)-\bu_1(a_1,\cdot)$ (and analogously for $v_1$)
contain both strictly positive and strictly negative components, ruling out weak
dominance. 
Second, the two games are not equivalent in the sense of \cref{def:game_eq}.
In particular, there is no player-specific positive affine transformation that
maps $\bu_1(\cdot,\cdot)$ to $\bv_1(\cdot,\cdot)$; see
\cref{proof:learnability_br} for a detailed verification.
Finally, the two games are indistinguishable under best-response feedback.
For player~1 and each action $a_1 \in A_1$, define the best-response region
\[
\mathcal{R}(a_1)
=
\left\{
\bx \in \mathbb{R}^2
\;\middle|\;
\langle \bx,\bu_1(a_1,\cdot) \rangle
\ge
\langle \bx,\bu_1(a_1',\cdot) \rangle
\ \forall a_1' \in A_1
\right\},
\]
and analogously for $v_1$. A direct comparison shows that the regions $\mathcal{R}(a_1)$ coincide for the two games for every $a_1 \in A_1$; see
\cref{proof:learnability_br} for a sketch of the best-response regions. Thus, for every recommendation mechanism, the induced best responses are identical
in the two games.

\end{proof}

\subsubsection{Characterization of the indistinguishability set under the best-response feedback}
As established above, the set of games indistinguishable under the best-response feedback can be strictly larger than the set of equivalent games characterized in \cref{def:game_eq}. For a given game $G(\{A_i\}_{i=1}^n, \{u_i\}_{i=1}^n)$, we provide a complete characterization of the set of games with no weakly dominated actions that are indistinguishable from $G$ under best-response feedback; we denote this set by $\mathcal{U}^{BR}(G)$. This characterization relies on the polyhedral structure induced by the best-response regions, which can be naturally described through normal fans of utility polytopes.
%
%

For a given agent $i$, consider the polytope $\mathcal P_i=\operatorname{conv}(\{\bu_i(a_i,\cdot)\}_{a_i\in A_i})$ $\subset \mathbb{R}^{d_i}$, where $\operatorname{conv}(\cdot)$ is the convex hull. For any face $F$ of $\mathcal P_i$, define its normal cone as $\mathcal N_{\mathcal P_i}(F)=\{y\in\mathbb R^{d_i}\mid \langle y,x\rangle\ge \langle y,x'\rangle \ \forall x\in F,\ \forall x'\in\mathcal P_i\}$. The normal fan of $\mathcal P_i$ is then $\mathcal N(\mathcal P_i)=\{\mathcal N_{\mathcal P_i}(F)\mid F\text{ is a face of }\mathcal P_i\}$. 
\edit{Further, let $C_i = \mathbb{R}^{d_i}_+$ be the full dimensional cone representing the positive orthant in $\mathbb{R}^{d_i}$}. In \cref{lem:Zero_char}, we give an initial characterization of the set $\mathcal{U}^{BR}$ in terms of the induced normal fan. 

\edit{
\begin{remark}
    In our setting, best-response feedback is observed only for recommendations $\bx \in \Delta(A)$. However, because the underlying geometry is invariant to positive rescaling of $\bx$, we may equivalently carry out the analysis over the positive orthant $\mathbb{R}^M_{+}$. Accordingly, for each agent $i$, we work with the cone $C_i := \mathbb{R}^{d_i}_{+}$. More generally, our characterization extends verbatim when $\bx$ is restricted to any other full-dimensional cone in $\mathbb{R}^{d_i}$.

\end{remark}
}

\begin{lemma}\label{lem:Zero_char}
Let $G(\{A_i\}_{i=1}^n,\{u_i\}_{i=1}^n)$ and $G'(\{A_i\}_{i=1}^n,\{v_i\}_{i=1}^n)$ be two generic games, and for each agent $i$ let $\mathcal P_i=\operatorname{conv}(\{\bu_i(a_i,\cdot)\}_{a_i\in A_i})$ and $\mathcal P'_i=\operatorname{conv}(\{\bv_i(a_i,\cdot)\}_{a_i\in A_i})$. We denote by 
\[
\mathcal N(\mathcal P_i)\mid_{C_i} \coloneqq \{\mathcal N_{\mathcal P_i}(F) \cap C_i\mid F\text{ is a face of }\mathcal P_i\}
\] 
where $\mathcal N_{\mathcal P_i}(F)$ is the normal cone  {of} $F$.  
Then the two games are indistinguishable under best-response feedback if and only if  
\[
\mathcal N(\mathcal P_i)\mid_{C_i}=\mathcal N(\mathcal P'_i)\mid_{C_i} \quad \forall i\in[n].
\]
\end{lemma}
\begin{proof}
    The proof is in \cref{proof:Zero_char}.
\end{proof}

Unlike a global normal fan  {of $\mathcal{P}_i$}, its restriction to $C_i$ leaves considerable freedom in the transformations allowed for the normal fan of the utility polytopes outside of $C_i$. Instead, we aim to reduce the characterization to comparing two \textit{global} normal fans of equivalent polyhedra. To this end, given a polyhedral cone $C_i$, let $C_{i}^\circ$ denote its polar cone and define the  {\textit{polarized polyhedron}} $\tilde{\mathcal P}_i := \mathcal P_i + C_{i}^\circ$  where $+$ denotes the Minkowski sum. Given two sets $\mathcal{C, K} \subseteq \mathbb{R}^d$, their Minkowski sum is defined as $\mathcal{C+K}  = \{x + y : x \in \mathcal{C},\, y \in \mathcal{K}\}$. 
 {\cref{prop:global_after_restriction} below uses this polyhedral construction to generalize the characterization.} 
\begin{proposition}
\label{prop:global_after_restriction}
Let $\mathcal P,\mathcal Q\subset\mathbb R^d$ be two polytopes and let $C\subset\mathbb R^d$ be a closed full-dimensional polyhedral cone. Suppose that $\mathcal N(\mathcal P)\mid_C=\mathcal N(\mathcal Q)\mid_C$.
Define the polarized polyhedra $\tilde{\mathcal P}:=\mathcal P+C^\circ$ and $\tilde{\mathcal Q}:=\mathcal Q+C^\circ$.
Then $\tilde{\mathcal P}$ and $\tilde{\mathcal Q}$ are normally equivalent, i.e., $\mathcal N(\tilde{\mathcal P})=\mathcal N(\tilde{\mathcal Q})$.
Conversely, if $\mathcal N(\tilde{\mathcal P})=\mathcal N(\tilde{\mathcal Q})$, then $\mathcal N(\mathcal P)\mid_C=\mathcal N(\mathcal Q)\mid_C$.
\end{proposition}
\begin{proof}
    The proof is in \cref{proof:global_after_restriction}.
\end{proof}

We defer the explicit polyhedral description of normal-equivalence, expressed as offset perturbations in a fixed $\mathcal H$-representation, to \cref{proof:h_rep_eq}. 

We now have all the ingredients needed to give a complete geometric characterization of the best-response indistinguishability set $\mathcal U^{BR}$.
Under our representation, a game is described by a collection of utility polytopes
\[
\mathcal P_i=\operatorname{conv}\bigl(\{u_i(a_i,\cdot)\}_{a_i\in A_i}\bigr), \qquad i\in[n].
\]
Under the standing assumption that there are no dominated strategies, this representation is essentially unique at the level of labeled actions. The non-dominance assumption ensures that for each agent $i$, every action $a_i$ corresponds to an exposed extreme point of $\mathcal P_i$, and conversely each vertex of $\mathcal P_i$ represents a unique action. However, relabeling the actions permutes the vertices of the corresponding polytope. While this changes the labeled description of the game, it does not affect the underlying (unlabeled) polytope. Throughout our analysis, we work in the labeled setting. With this, we can now state the geometric characterization of $\mathcal U^{BR}$ given in \cref{thm:ubr_characterization}.
The proof relies on \cref{lem:ext_points_polarization}.


\begin{theorem}[Geometric characterization of $\mathcal U^{BR}$]\label{thm:ubr_characterization}
Consider the game $G(\{A_i\}_{i=1}^n, \{u_i\}_{i=1}^n)$. For each agent $i\in[n]$, let $\mathcal P_i\subset\mathbb R^{d_i}$ be the associated utility polytope and define the polarized polyhedron
\[
\tilde{\mathcal P}_i := \mathcal P_i + C_i^\circ .
\]
Then the best-response indistinguishability set satisfies
\[
\mathcal U^{BR} (G) 
= 
\Bigl\{G'\bigl(\{A_i\}_{i=1}^n, \{\operatorname{\sigma}(\operatorname{ext}(Q_i))\}_{i=1}^n \big)\ \Big|\ 
\mathcal N(Q_i)=\mathcal N(\tilde{\mathcal P}_i)\ \ \forall i\in[n]\Bigr\},
\]
where $\operatorname{ext}(Q_i)$ denotes the set of extreme points of the polyhedron $Q_i$ and $\operatorname{sigma}(\operatorname{ext}(Q_i))$ labels the extreme point of Q to correspond to the labels of the extreme point of $\mathcal P$.
\end{theorem}

\vspace{-0.2cm}
\begin{proof}
Let $\mathcal P_i$ and $\mathcal P_i'$ be two utility polytopes of agent $i$ in two different games. 
By Lemma~\ref{lem:Zero_char}, they are indistinguishable under best-response feedback if and only if 
$\mathcal N(\mathcal P_i)\mid_{C_i}=\mathcal N(\mathcal P_i')\mid_{C_i}$. 
By Proposition~\ref{prop:global_after_restriction}, this implies 
$\mathcal N(\tilde{\mathcal P}_i)=\mathcal N(\tilde{\mathcal P}_i')$, where 
$\tilde{\mathcal P}_i=\mathcal P_i+C_i^\circ$ and $\tilde{\mathcal P}_i'=\mathcal P_i'+C_i^\circ$, 
hence $\tilde{\mathcal P}_i$ and $\tilde{\mathcal P}_i'$ are normally equivalent polyhedra.

Under the non-dominance (genericity) assumption, for every vertex 
$v\in\operatorname{ext}(\mathcal P_i')$ there exists $x\in C_i\setminus\{0\}$ such that 
$\langle x,v\rangle>\langle x,w\rangle$ for all 
$w\in\operatorname{ext}(\mathcal P_i')\setminus\{v\}$. 
Since these inequalities are strict and $\operatorname{ext}(\mathcal P_i')$ is finite, they are preserved under a small perturbation of $x$ inside $C_i$, yielding 
$y\in\operatorname{int}(C_i)$ with $\operatorname{face}_{\mathcal P_i'}(y)=\{v\}$. 
Thus the hypothesis of Lemma~\ref{lem:ext_points_polarization}(ii) holds and 
$\operatorname{ext}(\tilde{\mathcal P}_i')=\operatorname{ext}(\mathcal P_i')$. 
Therefore $\mathcal U^{BR}$ is contained in the set of games with utilities represented by the polytopes  
$(\operatorname{conv}(\operatorname{ext}(Q_i)))_{i\in[n]}$ such that 
$\mathcal N(Q_i)=\mathcal N(\tilde{\mathcal P}_i)$ for all $i$.

For the reverse inclusion, let $Q_i$ be normally equivalent to 
$\tilde{\mathcal P}_i=\mathcal P_i+C_i^\circ$. 
By the Minkowski--Weyl theorem, 
$Q_i=\operatorname{conv}(\operatorname{ext}(Q_i))+C_i^\circ$. 
Setting $\mathcal P_i':=\operatorname{conv}(\operatorname{ext}(Q_i))$ gives 
$Q_i=\tilde{\mathcal P}_i'$, and 
$\mathcal N(\tilde{\mathcal P}_i')=\mathcal N(\tilde{\mathcal P}_i)$ implies 
$\mathcal N(\mathcal P_i')\mid_{C_i}=\mathcal N(\mathcal P_i)\mid_{C_i}$, hence the game represented by the polytopes 
$(\mathcal P_i')_{i\in[n]}$ is included in $\mathcal U^{BR}$.
\end{proof}

\section{Algorithm for learning agent utilities from quantal-response feedback}
\label{sec:learning_alg}
As discussed in \cref{thm:learnability_qr}, a game $G(\{A_i\}_{i=1}^n, \{u_i\}_{i=1}^n)$ is learnable under the quantal-response feedback as per \cref{def:learnability_qr} of learnability. This result can be translated into a constructive algorithm that learns the set of all games equivalent to the true game through a series of recommendations $\bx$. Recall that the set of equivalent games is characterized by the positive affine transformation $\lambda_i \bu_i(a_i, \cdot) + \boldsymbol{t}_i$, where $\lambda_i > 0$ is an agent-specific scaling and $\boldsymbol{t}_i \in \mathbb{R}^{d_i}$ is an agent-specific translation.

\edit{Our algorithm learns the utility gain vectors $\bw_i(a_i,a_i') = \bu_i(a_i',\cdot) - \bu_i(a_i,\cdot)$ up to scaling, from which utility vectors equivalent to $\bu_i(a_i,\cdot)$ can be recovered.} It proceeds in three stages: (i) learning the \emph{sign patterns} of $\bw_i(a_i,a_i')$ (Algorithm \ref{alg:sign_pattern}), (ii) learning $\bw_i(a_i, a_i')$ up to a scaling factor (Algorithm \ref{alg:binary_ratio} and \ref{alg:recover_w}), and (iii) recovering their relative scales by solving a sparse linear program.  

\paragraph{Learning sign pattern.}  The sign pattern for a vector $\bw_{i}(a_i, a_i')$ is characterized by the index sets 
\[
\mathcal{P}_{i}(a_i, a_i') = \{k \ | \ (\bw_{i}(a_i, a_i'))_k \geq 0 \}, \quad \mathcal{N}_{i}(a_i, a_i') = \{k \ | \ (\bw_{i}(a_i, a_i'))_k < 0 \}
\]
Fixing an agent $i$, we can learn the sign pattern of all the vectors $\bw_{i}(a_i, a_i')$ for all $a_i'\neq a_i\in A_i$ with $d_i = |A_{-i}|$ recommendations mechanisms $\{\bx^{(k)}\}_{k=1}^{d_i}$ using \cref{alg:sign_pattern}, see \cref{lem:sign pattern} and its proof in \cref{proof:sign pattern}. 

\begin{algorithm}[h]
\caption{Learning the sign patterns of $\bw_i(a_i,a_i')$ for $a_i \neq a_i' \in A_i$}
\label{alg:sign_pattern}
\KwIn{Agent $i$, action set $A_i$}
\KwOut{Sign pattern sets $\mathcal{P}_i(a_i,a_i')$ and $\mathcal{N}_i(a_i,a_i')$ for all $a_i \neq a_i' \in A_i$}
\For{$k = 1, \ldots, d_i$}{
    Select a distribution $\bx^{(k)}$ such that
    \[
    \bx^{(k)}(a_i,\cdot) = \frac{1}{m_{i}} \be_k \in \mathbb{R}^{d}
    \quad \forall a_i \in A_i
    \]
    Recommend $\bx^{(k)}$ and collect quantal-response feedback sets
    $QR_i(\bx^{(k)}, a_i)$ for all $a_i \in A_i$\;

    \For{$a_i \in A_i$}{
        \For{$a_i' \in A_i \setminus \{a_i\}$}{
            \eIf{$a_i' \in QR_i(\bx^{(k)}, a_i)$}{
                $\mathcal{P}_i(a_i,a_i') \gets \mathcal{P}_i(a_i,a_i') \cup \{k\}$\;
            }{
                $\mathcal{N}_i(a_i,a_i') \gets \mathcal{N}_i(a_i,a_i') \cup \{k\}$\;
            }
        }
    }
}
\end{algorithm}

\paragraph{Learning $\bw_i(a_i, a_i')$ up to a scaling factor.} Given the learned sign pattern of $\bw_i(a_i, a_i')$, we select and normalize one positive coordinate of $\bw_i(a_i, a_i')$, and leverage it as a pivot to learn the relative magnitude of all negative components in  
$\bw_i(a_i, a_i')$ by constructing a sequence of recommendations as described in Algorithm~\ref{alg:binary_ratio}. This construction of recommendation sequence can be viewed as a binary search process. Moreover, to recover the positive components, we fix one of the learned negative components as pivot and repeat the binary search process for all the positive components. This whole process is described in Algorithm~\ref{alg:recover_w}, which allows us to learn the relative magnitude of all components in $\bw_i(a_i, a_i')$.

\begin{algorithm}[h]
\caption{Binary search with quantal-response feedback}
\label{alg:binary_ratio}
\KwIn{indices $(p,j)$ with $p \in \mathcal P_i(a_i,a_i')$, $j \in \mathcal N_i(a_i,a_i')$}
\KwOut{Approximate ratio $\tau^\star = - (\bw_i(a_i,a_i'))_p / (\bw_i(a_i,a_i'))_j$}

$a \leftarrow 0$, $b \leftarrow C$\;

\While{$b - a > \epsilon$}{
    $\tau \leftarrow (a+b)/2$\;
    
    Construct $\tilde{\bx} \in \Delta(A)$ such that
    \[
        \tilde{\bx}(a_i,\cdot) \leftarrow \be_p + \tau \be_j
    \]
    
    Normalize $\bx \leftarrow \tilde{\bx}/\|\tilde{\bx}\|_1$\;
    
    
    \eIf{$a_i' \in QR_i(\bx,a_i)$ \footnotemark}{
        $a \leftarrow \tau$\;
    }{
        $b \leftarrow \tau$\;
    }
}
\Return $(a+b)/2$\;
\end{algorithm}
\footnotetext{To check whether $a_i' \in QR_i(\bx,a_i)$, the moderator recommends $a_i$ and samples deviations until $a_i'$ is observed (or enough samples are collected to rule it out with high confidence); see \cref{lem:sampling_complexity}}

\begin{algorithm}[h!]
\caption{Learning $\bw_i(a_i,a_i')$ up to scaling}
\label{alg:recover_w}
\DontPrintSemicolon
\KwIn{Agent $i$, actions $a_i \neq a_i'$, sign patterns $\mathcal{P}_i(a_i, a_i')$ and $\mathcal{N}_i(a_i, a_i')$}
\KwOut{Estimate $\hat{\bw}_i(a_i,a_i') \in \mathbb{R}^{d_i}$}
Select pivots $p \in \mathcal{P}_i(a_i, a_i')$ and $q \in \mathcal{N}_i(a_i, a_i')$\;
Set $(\hat{\bw}_i(a_i,a_i'))_p \leftarrow 1$\;
\For{$j \in \mathcal{N}_i(a_i, a_i')$}{
    $\hat{\tau}_j \leftarrow$ Algorithm~\ref{alg:binary_ratio}$(p, j)$\;
    $(\hat{\bw}_i(a_i,a_i'))_j \leftarrow -1 / \hat{\tau}_j$\;
}
\For{$k \in \mathcal{P}_i(a_i, a_i') \setminus p$}{
    $ \hat{\tau}_k \leftarrow$ Algorithm~\ref{alg:binary_ratio}$(k, q)$\;
    $(\hat{\bw}_i(a_i,a_i'))_k \leftarrow - \hat w_q\, \hat{\tau}_k$\;
}
\Return $\hat{\bw}_i(a_i,a_i')$\;
\end{algorithm}

In \cref{thm:complexity_learning} we characterize the complexity, in terms of the number of recommendation mechanisms $\bx$, of recovering all vectors $\bw_i{(a_i, a_i')}$ for every agent $i \in [n]$ and action pair $a_i \neq a_i' \in A_i$ up to scaling. \edit{In \cref{lem:sampling_complexity} we give a standard high probability bound on the number of samples needed from each recommendation mechanism (that is the distribution $\bx$ in \cref{alg:binary_ratio}) to verify if a deviation $a_i'$ is in the qunatal-response set $QR_{i}(\bx, a_i)$}.

\begin{theorem}
\label{thm:complexity_learning}
Let $G(\{A_i\}_{i=1}^n,\{u_i\}_{i=1}^n)$ be a generic normal-form game with no weakly dominated actions. 
Assume that all utility differences are uniformly bounded away from zero and infinity, i.e.,
\[
c \;\le\; \big|u_i(a_i',\ba_{-i}) - u_i(a_i,\ba_{-i})\big| \;\le\; C
\]
for some constants $0 < c < C$, for all agents $i \in [n]$, action pairs $a_i \neq a_i' \in A_i$, and profiles $\ba_{-i} \in A_{-i}$.
Then, for any $\epsilon > 0$, the utility difference vectors $\bw_i(a_i,a_i')$ can be learned up to a positive scaling, with additive error $O(\epsilon)$, using
\[
\edit{O\!\left(nm M \log(1/\epsilon)\right)}
\]
recommendation mechanisms by applying Algorithms~\ref{alg:sign_pattern} and~\ref{alg:recover_w}.
\end{theorem}


\begin{proof}
Fix an agent $i \in [n]$.

\paragraph{Step 1: Learning sign patterns.}
By algorithm \cref{alg:sign_pattern}, we can learn the sign pattern sets $\mathcal P_i(a_i,a_i')$ and $\mathcal N_i(a_i,a_i')$ for every action pair $a_i \neq a_i' \in A_i$ using exactly $d_i \coloneqq |A_{-i}|$ recommendation mechanisms (cf. \cref{lem:sign pattern}).

\paragraph{Step 2: Learning $\bw_i(a_i,a_i')$ for a fixed $a_i \neq a_i'$ up to scaling.}
Fix a pair $a_i \neq a_i'$ and denote
\[
\bw \coloneqq \bw_i(a_i,a_i') \in \mathbb{R}^{d_i}.
\]
Since the game has no weakly dominated actions, both $\mathcal P_i(a_i,a_i')$ and $\mathcal N_i(a_i,a_i')$ are nonempty. Fix pivot indices $p \in \mathcal P_i(a_i,a_i')$ and $q \in \mathcal N_i(a_i,a_i')$. We fix the scaling by targeting the normalized vector $\bw/w_p$ (equivalently, set the pivot to $\hat w_p:=1$). 

For any $j \in \mathcal N_i(a_i,a_i')$, querying the recommendation mechanism $\bx^{(\tau)}$ such that 
\[
\bx^{(\tau)}(a_i, \cdot) \;\coloneqq\; \frac{\be_p + \tau \be_j}{1+\tau},
\qquad \be_p, \be_j \in \mathbb{R}^{d_i}, \ \tau \ge 0, 
\]
and observing the qunatal-response feedback set $QR_{i}(\bx^{(\tau)}, a_i)$ reveals
\[
\operatorname{sign}\langle \bx^{(\tau)}(a_i, \cdot), \bw \rangle = \operatorname{sign}(w_p + \tau w_j).
\]
Since $w_p > 0$ and $w_j < 0$, $w_p + \tau w_j$ is a strictly decreasing function with a unique root $\tau^{\star}_j = -w_p/w_j$ which is well-defined since $|w_p|, |w_j| > c$. Therefore, using binary search on the value of $\tau^{\star}_j$ with sign feedback we recover an estimate $\hat{\tau}_j$ such that $|\hat\tau_j - \tau_j^\star| \le \epsilon$. We set
\[
\hat w_j \;\coloneqq\; -\frac{1}{\hat\tau_j},
\]
which approximates $w_j / w_p$ with additive error $O(\epsilon)$.

\edit{
For any $k \in \mathcal P_i(a_i,a_i')$, repeating the same root-finding procedure recovers an estimate $\hat{\tau}_k$ such that $|\hat{\tau}_k - \tau^{\star}_k| \leq \epsilon$, where $\tau^{\star}_k = -w_k/w_q$. Since the negative pivot $\hat w_q$ is only estimated up to $O(\epsilon)$ accuracy, we estimate $\hat{w}_k = -\hat{\tau}_k \hat{w}_q$. The error in estimating $\hat{w}_k$ propagates as 
\begin{align*}
    \big |\hat{w}_k - \frac{w_k}{w_p}\big| & \leq \big |-\hat{\tau}_k\hat{w}_q + \tau_k^{\star} \hat w_q - \tau_k^{\star} \hat w_q + \tau_k^{\star} \frac{w_q}{w_p} \big| \\
    & \leq 
    \hat{w}_q |-\hat{\tau}_k + \tau_k^{\star} \big | +  \tau_k^{\star}  \big | \hat w_q - \frac{w_q}{w_p} \big|\\
    & 
    \leq \hat{w}_q \epsilon + \tau_{k}^{\star} O(\epsilon) \\
    & \leq C \epsilon + \frac{C}{c} O(\epsilon) \\
    & = O(\epsilon).
\end{align*}
}

Finally, each coordinate of $\bw$ can be learned by a binary search procedure requiring $O(\log(1/\epsilon))$ recommendations mechanisms. 

\paragraph{Step 3: Total number of calls to \cref{alg:binary_ratio}.}
\edit{Each utility difference vector $\bw_i(a_i, a_i') \in \mathbb{R}^{d_i}$ is recovered up to scale using $d_i-1$ calls to \cref{alg:binary_ratio}. Each agent $i$ has $m_i(m_i - 1)/2$ such vectors, so the number of calls for agent $i$ is $O(d_i m_i(m_i-1)) = O(M(m_i-1))$, where $M= |A|$. Summing over $i \in [n]$, gives $O(M\sum_{i\in [n]}(m_i-1)) \leq O(nmM)$, where $m = \max_{i\in [n]} m_i$. So the total number of recommendation mechanisms required to learn (up to scale) all utility vectors $\bw_{i}(a_i, a_i')$ for all $i \in [n]$, and $a_i \neq a_i' \in A_i$ is $O\big(nmM \log(1/\epsilon)\big)$, as claimed. }

\end{proof}
\edit{
\begin{lemma}
\label{lem:sampling_complexity}
Consider the game $G(\{A_i\}_{i=1}^n,\{u_i\}_{i=1}^n)$ and assume that all utility differences are uniformly bounded away from zero and infinity, 
\(
c \;\le\; \big|u_i(a_i',\ba_{-i}) - u_i(a_i,\ba_{-i})\big| \;\le\; C 
\)
for some constants $0 < c < C$, for all agents $i \in [n]$, action pairs $a_i \neq a_i' \in A_i$, and profiles $\ba_{-i} \in A_{-i}$. Fix an agent $i$ and a pair of actions $a_i \neq a_i'$. Then for each recommendation mechanism $\bx$ constructed as in \cref{alg:binary_ratio} the moderator needs at most 
\(m_i e^{\beta C} \ln(1/\delta)\)
samples to verify if $a_i' \in QR_i(\bx, a_i)$ with probability at least $1-\delta$. 
\end{lemma}
\begin{proof}
    The proof is in \cref{proof:sampling_complexity}
\end{proof}
\begin{remark}
\Cref{lem:sampling_complexity} highlights how agents' bounded rationality controls the learnability of the utility functions. When the rationality parameter $\beta = 0$ (maximally noisy behavior), the moderator requires on the order of $m_i \ln(1/\delta)$ samples from each recommendation mechanism to identify the utilities up to the relevant equivalence class. As $\beta \to \infty$ and agents become perfectly rational, the feedback exponentially converges to best-response behavior, recovering the corresponding impossibility of learning.
\end{remark}
}
\paragraph{Identifying the relative scale.}
So far, we have recovered every utility difference vector $\bw_i(a_i,a_i')$ up to an independent positive scaling. These scalings cannot be arbitrary, since these vectors must satisfy the triangular identities
\[
\bw_i(a,c)=\bw_i(a,b)+\bw_i(b,c)\quad \forall a,b,c\in A_i .
\]
which translate for the scaled estimates into the consistency relations
\(
\lambda_{ac}\hat w_i(a,c)=\lambda_{ab}\hat w_i(a,b)+\lambda_{bc}\hat w_i(b,c)
.\) These consistency relations define a sparse linear system in the scaling variables $\lambda$, where each constraint involves only three unknowns. Such systems are classically solved by projection-based methods (e.g., Kaczmarz-type algorithms), stochastic least-squares schemes, or sparse iterative solvers from numerical linear algebra.

\section{Algorithm for low-regret recommendations} 
\label{sec:low_reg_alg}
In this section, we present an online algorithm that enables a moderator to generate a sequence of low-regret recommendations from action feedback under the best-response and quantal-response feedback models. The algorithm is based on a cutting-plane method that iteratively searches over an unknown parameter vector $\bw^{\star} \in \mathbb{R}^{N}$, where $N = \sum_{i\in [n]} (m_i-1)|A_{-i}|$. This vector is formed by stacking the pairwise utility difference vectors $\bw_i(a_i,a_i')$ without redundancy. For each agent $i$, fix a reference action $a_i^1$ and write 
\[
\bw^\star
=
\begin{bmatrix}
\bw_1^\star \\
\bw_2^\star \\
\vdots \\
\bw_n^\star
\end{bmatrix},
\qquad
\bw_i^\star
=
\begin{bmatrix}
\bw_i(a_i^1, a_i^j)
\end{bmatrix}_{j \in [2, \dots, m_i]}.
\]
Note that $\bw_i(a_i^j, a_i^1) = - \bw_i(a_i^1, a_i^j)$ for each $j\in [2, \dots,m_i]$. Further, the utility difference vector for any pair of actions $a', a'' \in A_i \setminus \{a_i^1\}$ can be recovered from $\bw^\star$ via the triangular identity
\[
\bw_i(a', a'') \;=\; \bw_i(a_i^1, a'') \;-\; \bw_i(a_i^1, a').
\]
We assume that $\bw^{\star} \in  \mathcal{B} \triangleq \{\bw \in \mathbb{R}^N\ \big|\ \|\bw\|\leq 1 \}$. This is without loss of generality as scaling the vector generates an equivalent game for both feedback models.

It will be useful to first outline the basic structure of a cutting-plane algorithm. At each iteration $t$, we select a query point $\bw^{(t)}$ and submit it to a \textit{separation oracle}. The oracle returns a direction $\bq^{(t)}$ defining a half-space such that $\langle \bw^{\star}, \bq^{(t)} \rangle \geq 0$ and $\langle \bw^{(t)}, \bq^{(t)}\rangle \leq 0$. We also maintain a \textit{knowledge set} $\mathcal{C}^{(t)}$ of all vectors consistent with the information gathered up to iteration $t$. We initialize $\mathcal{C}^{(0)} = \mathcal{B}$, and using the oracle’s response, we update $\mathcal{C}^{(t)} = \mathcal{C}^{(t-1)} \cap \{\bw \in \mathcal{B} \ \big| \ \langle\bw, \bq^{(t)}\rangle \geq 0 \}$. 

We also introduce preliminary notation from convex geometry that will be useful throughout the discussion. Let \(\mathcal{C} \subset \mathbb{R}^N\) be a convex set. We denote its volume by $\mathrm{Vol}(\mathcal{C})$ and its center of gravity by $\mathrm{cg}(\mathcal{C})$. The width of $\mathcal{C}$ in a direction $\bv \in \mathbb{R}^N$ is $\text{width}(\mathcal{C}, \bv) = \max_{\boldsymbol{y}, \bz \in \mathcal{C}} \langle \boldsymbol{y}-\bz, \bv\rangle$.  

\subsection{Low-regret recommendations}
\Cref{alg:low_regret_rec} below describes how to use cutting-plane as a method to construct a sequence of low-regret recommendations. It mainly describes how to select the query points $\bw^{(t)}$, generate recommendation mechanisms $\bx^{(t)}$ based on these queries, and use the feedback from the the recommendations to construct a separation oracle. This construction is the key to relating the geometric progress of the cutting-plane method to the regret accumulated by the moderator.

\paragraph{Selecting the query points.} In every iteration $t = 1, \dots, T$, the algorithm select a query point 
\begin{align}
\label{eq:query_selection}
    \bw^{(t)} =  \mathrm{cg}\!\left(\mathcal{C}^{t-1} + \frac{1}{T}\mathcal{B}\right), 
\end{align}
where $\mathcal{C}^{(t-1)}$ is the current knowledge set and $\mathcal{B}$ is the unit ball in $\mathbb{R}^{N}$. In \cref{sec:proof_low_regret}, we justify the choice of the query points by showing that selecting the center of gravity of the buffered set  $\mathrm{cg}\!\left(\mathcal{C}^{t-1} + \frac{1}{T}\mathcal{B}\right)$ rather than the center of gravity of $\mathcal{C}^{t-1}$ itself allows us to reduce the \textit{width} of the knowledge set every time the algorithm incurs a large regret.

\paragraph{Generating recommendations.} Note that the normalized vector $\bw^{\star}$ encodes all the information required to compute 
a correlated equilibrium of the underlying game. So each query point $\bw^{(t)}$ produced by the algorithm can therefore be interpreted as a guess of the agents' utilities in the game. For any such game $\bw^{(t)}$, the set of correlated equilibria, denoted $CE(\bw^{(t)})$, is nonempty \cite{hart1989existence}. Given a point $\bw^{(t)}$, the algorithm constructs the recommendation mechanism $\bx^{(t)} \in CE(\bw^{(t)})$, recommends actions $\ba^{(t)} \sim \bx^{(t)}$, and observes the realized action $\ba^{\star (t)}$ based on one of the two feedback models. Upon receiving the feedback $\ba^{\star (t)}$, we consider one of the following two cases: 
\begin{itemize}
    \item[-] \textbf{Case 1: } If $\ba^{(t)} = \ba^{\star (t)}$, then no deviations are observed and the algorithm incurs zero regret $r(\ba^{(t)}, \ba^{\star (t)}, \bx^{(t)}) = 0$. In this case \edit{we do not update $x^{(t)}$} and sample another $\ba^{(t+1)} \sim \bx^{(t)}$. 
     \item[-] \textbf{Case 2:} If $\ba^{(t)} \neq \ba^{\star (t)}$, that is some agents deviated from the recommended action then we can construct a separating hyperplane $\bq^{(t)} \in \mathbb{R}^{N}$ such that $\langle \bw^{\star}, \bq^{(t)}\rangle \geq0$ and $\langle \bw^{(t)}, \bq^{(t)}\rangle \leq0$ as described in the next paragraph.
\end{itemize}
\begin{algorithm}[h!]
\caption{Separation oracle}
\label{alg:oracle}
\DontPrintSemicolon
\textbf{Input:} Action recommendations $\ba^{(t)}$ and realized actions $\ba^{\star(t)}$\;
\textbf{Output:} Separating hyperplane $\bq^{(t)} \in \mathbb{R}^{N}$\;

Initialize $\bq^{(t)} \gets \boldsymbol{0}$\;

\For{$i \leftarrow 1$ \KwTo $n$}{
    \If{$a_i^{(t)} \neq a_i^{\star(t)}$}{
        $j \gets \gamma(a_i^{(t)})$\;
        $j^\star \gets \gamma(a_i^{\star(t)})$\;

        \If{$j \neq 1$}{
            $\bq^{(t)}(i,j) \gets \bq^{(t)}(i,j) - \bx^{(t)}(a_i^j,\cdot)$\;
        }
        \If{$j^\star \neq 1$}{
            $\bq^{(t)}(i,j^\star) \gets \bq^{(t)}(i,j^\star) + \bx^{(t)}(a_i^j,\cdot)$\;
        }
    }
}
\end{algorithm}

\paragraph{Constructing the separation oracle.}
Recall that the vector $\bw^{\star}$ is constructed by concatenating $\bw_i(a_i^1, a_i^j)$ for $j = 2, \dotsm m_i$ and $i \in [n]$. Given an action $a_i$, let $\gamma(a_i)$ be the index of action $a_i$ in the action set $A_i$. For each $i \in [n]$ and $j \in [2, \dots, m_i]$, let $\ell(i, j) = (i-1)(m_i-1)d_i + (j-2)d_i +1$ be the starting index of the block corresponding to $\bw_i(a_i^1,a_i^{j})$ in $\bw^{\star}$. Denote by $\bq^{(t)}(i, j) = \bq[\ell(i, j): \ell(i, j) + d_i -1]$ the slice of the vector $\bq^{(t)}$ corresponding to the location of the vector $\bw_i(a_i^{1}, a_i^{j})$ in $\bw^{\star}$.
In \cref{alg:oracle}, we show how to use the action feedback $\ba^{\star (t)}$, collected as a response to the recommendation $\ba^{(t)} \sim \bx^{(t)}$, to construct a hyperplane $\bq^{(t)} \in \mathbb{R}^{N}$ that separates $\bw^{\star}$ and $\bw^{(t)}$. 

The oracle in \cref{alg:oracle} is constructed such that 
\begin{align}
\label{eq:oracle_positive}
    \langle \bw^{\star}, \bq^{(t)}\rangle =  \sum_{i\in [n]} \varphi_i(a_i^{(t)}, a_i^{\star (t)}, \bx^{(t)}),
\end{align}
where $\varphi(a_i^{(t)}, a_i^{\star (t)}, \bx^{(t)}) = \langle \bw_i (a_i^{(t)}, a_i^{\star (t)}), \bx^{(t)}(a_i^{(t)}, \cdot)\rangle$. This $\bq^{(t)}$ is a valid separating hyperplane for both feedback models, that is $\langle \bw^{\star}, \bq^{(t)} \rangle \geq 0$ and $\langle \bw^{(t)}, \bq^{(t)}\rangle \leq 0$. This follows form the fact that $\varphi(a_i^{(t)}, a_i^{\star (t)}, \bx^{(t)}) \geq 0$ for every $a_i^{\star (t)} \in BR_i(\bx^{(t)}, a_i^{(t)})$ (cf. \cref{def:best_response}) or $a_i^{\star (t)} \in QR_i(\bx^{(t)}, a_i^{(t)})$ (cf. \cref{def:quantal_response}). Finally, since  $\bx^{(t)} \in CE(\bw^{(t)})$ we have $\langle \bw^{(t)}, \bq^{(t)} \rangle \leq 0$ as desired.



    
    
    
    

\begin{algorithm}[h!]
\caption{Low-Regret Recommendations}
\label{alg:low_regret_rec}
\textbf{Input:}Baseline set $\mathcal{B}$, time horizon $T$\;
\textbf{Output:}Sequence of recommendations $\{\bx^{(t)}\}_{t=1}^T$\;

Initialize $\mathcal{C}^0 \gets \mathcal{B}$\;
Initialize $t\gets 1$\;
Initialize $\operatorname{UPDATE} \gets \operatorname{True}$\;

\For{$t  = 1, \dots, T$}{
    \If{$\operatorname{UPDATE}$}{
    Select query point $\bw^{(t)}$ as in \cref{eq:query_selection}\;
    Compute a recommendation $\bx^{(t)} \in CE(\bw^{(t)})$\;
    }

    \Else{
    $\bx^{(t)} \gets \bx^{(t-1)}$ \;
    }

    Sample action recommendation $\ba^{(t)} \sim \bx^{(t)}$ and collect feedback $\ba^{\star (t)}$\;
    
    \If{$\ba^{\star (t)} \neq \ba^{(t)}$}{
    Construct $\bq^{(t)}$ as in \cref{alg:oracle}\;
    
    Update the knowledge set $\mathcal{C}^t \gets \mathcal{C}^{t-1} \cap \{\bw \in \mathbb{R}^N \mid \langle \bw, \bq^{(t)} \rangle \ge 0\}$\;
    
    $\operatorname{UPDATE} \gets \operatorname{True}$\;
    }
    \Else{$\operatorname{UPDATE} \gets \operatorname{False}$ }
}

\end{algorithm}

\paragraph{Regret bound.} 
In \cref{thm:low-regret} we give a bound on the expected cumulative regret of the sequence of recommendations generated by \cref{alg:low_regret_rec}. The bound holds when agents follow both the best-response or the quantal-response feedback models.
We show that the incurred regret scales linearly with \edit{the size of the game representation $nM$} and logarithmic with the number of iterations $T$. The proof of the theorem is provided in \cref{sec:proof_low_regret}. 

    \begin{theorem}
    \label{thm:low-regret} The recommendations generated by \cref{alg:low_regret_rec} incurs regret \edit{$O\big(nM\operatorname{log}(T)\big)$}.
    \end{theorem}

\subsection{Proof of \cref{thm:low-regret}} 
\label{sec:proof_low_regret}
The proof of \cref{thm:low-regret} involves two main steps. First, we show that selecting query points as in \cref{eq:query_selection} guarantees progress in shrinking the \textit{width} of the knowledge set. Second, we relate this progress to the regret of the recommendations via the construction of the separation oracle.

Classical cutting-plane methods aim to remove a constant fraction of the volume of the knowledge set at every iteration. Such volume-reduction guarantees are insufficient for our setting, since our primary concern is not the volume of the set but rather its width. Reducing volume alone does not necessarily control the width: a convex set may have arbitrarily small volume while remaining highly elongated in some direction. To overcome this limitation, we draw on recent works on cutting-plane methods for contextual search \cite{gollapudi2021contextual,leme2018contextual,liu2021optimal,lobel2018multidimensional}, which control the width of the knowledge set using various geometric techniques. 

The key technique underlying the query selection in \cref{alg:low_regret_rec} is the maintenance of a potential function $\mathrm{Vol}(\mathcal{C} + \rho \mathcal{B})$, where $\rho> 0$ is a regularization parameter. With an appropriate choice of $\rho$, this potential links the volume of the expanded set $\mathcal{C} + \rho\mathcal{B}$ to the width of $\mathcal{C}$. Consider the distance between the query point $\bw^{(t)}$ and the true point $\bw^{\star}$ along the oracle direction 
\begin{align}
   D(\bw^{(t)}) \coloneqq \langle \bw^{\star} - \bw^{(t)}, \bq^{(t)}\rangle, 
\label{eq:cp_regret}
\end{align}
which is bounded by the width of the current set along $\bq^{(t)}$. When $\text{width}(\mathcal{C}^{(t)}, \bq^{(t)})$ is large relative to $\rho$, cutting the expanded set $\mathcal{C}^{(t)} + \rho\mathcal{B}$ through its centroid reduces $\text{Vol}(\mathcal{C}^{(t)}+\rho\mathcal{B})$ by a constant factor. Since the set $\mathcal{C}^{(t)}$ is never empty, we always have $\text{Vol}(\mathcal{C}^{(t)}+\rho\mathcal{B}) \ge \text{Vol}(\rho\mathcal{B})$. This implies that such large-width cases can occur only in a bounded number of times. We formalize this in \cref{lem:low_regret_cutting} and refer to \cite{gollapudi2021contextual}[Theorem 4.2] for a complete proof.

\begin{lemma}
\label{lem:low_regret_cutting}
The cumulative distance $\sum_{t=1}^T D(\bw^{(t)})$ of the points $\bw^{(t)}$ generated by a cutting-plane algorithm that queries $\bw^{(t)} = \mathrm{cg}(\mathcal{C}^{(t-1)} + \frac{1}{T}\mathcal{B})$ is $O(N\log(T))$. 
\end{lemma}

The second step to derive the regret bound in \cref{thm:low-regret} is to relate the progress of the cutting-plane to the regret of the recommendations generated by \cref{alg:low_regret_rec}. We show in \cref{lem:regret_of_cp_rec} that for each $t$, the regret of the recommendation $r(\ba^{(t)},\ba^{\star (t)}, \bx^{(t)})$ is bounded by the distance $D(\bw^{(t)})$ of the query point $\bw^{(t)}$ given the separation oracle constructed in \cref{alg:oracle}. 

\begin{lemma}
    Let $\bw^{(t)} \in \mathcal{B}$ and $\bx^{(t)} \in CE(\bw^{(t)})$. For $\bq^{(t)}$ constructed as in \cref{alg:oracle}, we have $$r(\ba^{(t)},\ba^{\star (t)}, \bx^{(t)}) \leq D(\bw^{(t)}).$$ 
\label{lem:regret_of_cp_rec}
\end{lemma}

\begin{proof}
    Recall the regret  $r(\ba^{(t)},\ba^{\star (t)}, \bx^{(t)}) =  \sum_{i \in [n]}  \varphi_i(a_i^{(t)}, a_i^{\star (t)}, \bx^{(t)})$
    By construction of $\bq^{(t)}$ we have 
    $$r(\ba^{(t)},\ba^{\star (t)}, \bx^{(t)}) = \langle \bw^{\star}, \bq^{(t)}\rangle.$$ 
    Further, the oracle guarantees $\langle \bw^{(t)}, \bq^{(t)}\rangle \leq 0$. Thus, we have  $D(\bw^{(t)}) = \langle \bw^{\star} - \bw^{(t)}, \bq^{(t)}\rangle \geq \langle \bw^{\star}, \bq^{(t)}\rangle = r(\ba^{(t)},\ba^{\star (t)}, \bx^{(t)})$. 
\end{proof}

The regret bound in \cref{thm:low-regret} follows by \cref{lem:low_regret_cutting,lem:regret_of_cp_rec} and summing over $t = 1, \dots, T$ \edit{
\begin{align}
    \sum_{t=1}^T r(\ba^{(t)},\ba^{\star (t)}, \bx^{(t)}) \leq \sum_{t=1}^T D(\bw^{(t)}) = O(N\operatorname{log}(T)) \leq O(nM \operatorname{log}(T)),
\end{align}
where the last inequality follows by $N <nM$. Finally, \cref{alg:low_regret_rec} is a low-regret algorithm} since the bound holds uniformly for every realized trajectory $\{(\ba^{(t)},\ba^{\star (t)}, \bx^{(t)})\}_{t=1}^T$. Taking the expectation preserves the bound so 
\begin{align}
    \mathbb{E} \left[\sum_{t=1}^T r(\ba^{(t)},\ba^{\star (t)}, \bx^{(t)})\right] = O(nM \operatorname{log}(T)).
\end{align}

\bibliographystyle{plain}
\bibliography{ref}

\appendix

\section{Omitted results and proofs from \cref{sec:learnability}}
\label{appx:omitted_proofs_learnability}
\subsection{Proof of \cref{lem:same_halfspace}}
\label{proof:same_halfspace}
\begin{lemma}
Let $\bw, \bw'\in\mathbb{R}^k$ be two vectors each having at least one strictly positive
and one strictly negative component.
If for every $\bx \in \mathbb{R}_{+}$ we have $\operatorname{sign}\langle \bx, \bw\rangle
= \operatorname{sign}\langle \bx,\bw'\rangle$, then there exists $\lambda > 0$ such that $\bw'=\lambda \bw$.
\end{lemma}
\begin{proof}

Let $\{\be_i\}_{i}^k$ be the canonical bases vectors of $\mathbb{R}^k$. For each coordinate $i \in [k]$ let $\bx=\be_i$. Then 
\begin{align*}
    \operatorname{sign}\langle \bx, \bw\rangle
= \operatorname{sign}\langle \bx,\bw'\rangle \implies \operatorname{sign}(w_i)=\operatorname{sign}(w'_i)
\end{align*}
Hence $\bw$ and $\bw'$ share the same \textit{sign pattern}.

Suppose $w_i=0$ but $w'_i \ne 0$. Pick index $j \in [d]$ such that $w_j$ of strictly opposite sign to $w'_i$ (such $j$ exists by the mixed-sign assumption). 
Consider $\bx=\be_i+t\be_j$ for $t>0$, then $\langle \bx,\bw\rangle=t\,w_j$ and $\langle \bx,\bw'\rangle=w'_i+t\,w'_j$. By varying $t$, one can flip the sign of $\langle \bx,\bw'\rangle$ while $\langle \bx,\bw\rangle$ keeps a fixed sign (the sign of $w_j$), contradicting the hypothesis. Thus $w_i=0\iff w'_i=0$.

Similarly, Let $i, j$ be indices such that $w_i$ and $w_j$ are of opposite signs and without loss of generality fix $w_i>0$, $w_j<0$ . Set $\bx^{(t)}= \be_i+t\,\be_j$ for $t\ge0$.  
Then
\[
\langle \bx^{(t)},\bw\rangle = w_i + t\,w_j, \qquad
\langle \bx^{(t)},\bw'\rangle = w'_i + t\,w'_j.
\]
The first expression changes sign at $t_0=-w_i/w_j>0$.
Since both inner products must share the same sign for all $t\ge0$,
the second one must vanish at the same $t_0$, giving
\[
\frac{w'_i}{w_i}=\frac{w'_j}{w_j}.
\]
Because the pair $i,j$ is arbitrary among positive/negative indices,
this ratio is constant over all coordinates.  
Hence $\bw'=\lambda \bw$ for some $\lambda>0$.
\end{proof}

\subsection{Detailed proof of \cref{thm:learnability_br}}
\label{proof:learnability_br}
Recall the two-player games $G(\{A_i\}_{i=1}^2, \{u_i\}_{i=1}^2)$ and $G(\{A_i\}_{i=1}^2, \{v_i\}_{i=1}^2)$ with action sets $A_1 = \{a^1, a^2, a^3, a^4\}$ and $A_2 = \{b^1, b^2\}$. Let the utilities of player 1 in these two games be
\begin{align*}
\small
u_1\ & = \ \begin{array}{c|cccc}
 & a^1 & a^2 & a^3 & a^4 \\
\hline
b^1 & 0 & 3 & 5 & 8 \\
b^2 & 8 & 6.5 & 4.5 & 0
\end{array}
\qquad
v_1 \  = \ \begin{array}{c|cccc}
 & a^1 & a^2 & a^3 & a^4 \\
\hline
b^1 & 0 & 2 & 6 & 8 \\
b^2 & 8 & 7 & 3 & 0
\end{array}
\end{align*}
And let the utilities of player 2 be identical across the two games. 

\begin{proposition}
    The games $G(\{A_i\}_{i=1}^2, \{u_i\}_{i=1}^2)$ and $G(\{A_i\}_{i=1}^2, \{v_i\}_{i=1}^2)$ are not equivalent. 
\end{proposition}
\begin{proof}
    For the two games to be equivalent there must exist a scalar $\lambda_1>0$ and vector $\boldsymbol{t} \in \mathbb{R}^{2}$ such that: 
    \begin{align*}
        \bv_1(a_1, \cdot) = \lambda_1 \bu_1(a_1, \cdot) + \boldsymbol{t}
    \end{align*}
    Taking action $a_1^1$ we have
    \begin{align*}
        \begin{bmatrix}
            0 \\ 8
        \end{bmatrix} = \lambda_1 \begin{bmatrix}
            0 \\ 8
        \end{bmatrix}  + \begin{bmatrix}
            t_1 \\ t_2
        \end{bmatrix}, 
    \end{align*}
    forcing $t_1 = 0$. Similarly, by taking actions $a_1^4$ we get 
    \begin{align*}
        \begin{bmatrix}
            8 \\ 0
        \end{bmatrix} = \lambda_1 \begin{bmatrix}
            8 \\ 0
        \end{bmatrix}  + \begin{bmatrix}
            t_1 \\ t_2
        \end{bmatrix}, 
    \end{align*}
    forcing $t_2 = 0$. Finally, there is no common scalar $\lambda_i$ such that $\bv_1(a_1, \cdot) = \lambda_1 \bu_1(a_1, \cdot)$ for all $a_1 \in A_1$ which proves the claim. 
\end{proof}



\begin{proposition}
    The games $G(\{A_i\}_{i=1}^2, \{u_i\}_{i=1}^2)$ and $G(\{A_i\}_{i=1}^2, \{v_i\}_{i=1}^2)$ are indistinguishable under the BR feedback.
\end{proposition}
\begin{proof}
    Construct the utility polytopes for player 1 in each game where the utility polytope of a player in a game is the convex hull of their utility vectors $\bu_{1}(a_1, \cdot)$ for each $a_1 \in A_1$ (cf. \cref{fig:polytopes}). We can prove that the two games are indistinguishable under the BR feedback geometrically by showing that they have the same best response regions for each action. We sketch these best-response regions in \cref{fig:polytopes} proving that they are identical. 
\end{proof}

\begin{figure}[t]
  \centering
\begin{tikzpicture}[
  scale=0.5,
  line cap=round, line join=round,
  every node/.style={inner sep=1pt},
  cone/.style={opacity=0.22},
  coneborder/.style={line width=0.9pt, dashed},
  poly/.style={fill=gray!10, draw=black, line width=0.9pt},
  vtx/.style={circle, inner sep=2.2pt}
]

\coordinate (x1) at (0,  8);
\coordinate (x2) at (3,  6.5);
\coordinate (x3) at (5,  4.5);
\coordinate (x4) at (8, 0);

\coordinate (x2') at (2, 7);
\coordinate (x3') at (6, 3);
\draw[poly] (x1) -- (x2) -- (x3) -- (x4) -- (x1) -- cycle;

\fill[red]              (x1) circle (2.6pt);
\fill[blue]             (x2) circle (2.6pt);
\fill[green!70!black]   (x3) circle (2.6pt);
\fill[orange!90!black]  (x4) circle (2.6pt);

\node[red, right=10pt, font=\footnotesize]  at (x1) {$\bu_{1}(a_1^{1}, \cdot)$};
\node[blue, right=10pt, font=\footnotesize]       at (x2) {$\bu_{1}(a_1^{1}, \cdot)$};
\node[green!70!black,   right=10pt, font=\footnotesize] at (x3) {$\bu_{1}(a_1^{3}, \cdot)$};
\node[orange!90!black,  left=10pt, font=\footnotesize] at (x4) {$\bu_{1}(a_1^{4}, \cdot)$};


\coordinate (x1a) at (0.75,9.5); 
\coordinate (x1b) at (-1, 7); 
\fill[red, cone] (x1) -- (x1a) -- (x1b);
\draw[red, coneborder] (x1a) -- (x1) -- (x1b);
\node[red, above left=1pt, font=\footnotesize]  at (x1) {$\mathcal{R}(a_1^1)$};

\coordinate (x2a) at ((4.5,8); 
\coordinate (x2b) at (4,8.5);
\fill[blue, cone] (x2) -- (x2a) -- (x2b);
\draw[blue, coneborder] (x2a) -- (x2) -- (x2b);
\node[blue, xshift=30pt, yshift=30pt, font=\footnotesize] at (x2) {$\mathcal{R}(a_1^2)$};

\coordinate (x3a) at (7.25, 6); 
\coordinate (x3b) at (7,6.5); 
\fill[green!70!black, cone] (x3) -- (x3a) -- (x3b) -- cycle;
\draw[green!70!black, coneborder] (x3a) -- (x3) -- (x3b);
\node[green!70!black, xshift=45pt, yshift=30pt, font=\footnotesize] at (x3) {$\mathcal{R}(a_1^3)$};

\coordinate (x4a) at (9.5,1); 
\coordinate (x4b) at (7, -1);
\fill[orange!90!black, cone] (x4) -- (x4a) -- (x4b) -- cycle;
\draw[orange!90!black, coneborder] (x4a) -- (x4) -- (x4b);
\node[orange!90!black, xshift=10pt, yshift=-10pt, font=\footnotesize] at (x4) {$\mathcal{R}(a_1^4)$};

\end{tikzpicture}
  \hspace{1cm}
\begin{tikzpicture}[
  scale=0.5,
  line cap=round, line join=round,
  every node/.style={inner sep=1pt},
  cone/.style={opacity=0.22},
  coneborder/.style={line width=0.9pt, dashed},
  poly/.style={fill=gray!10, draw=black, line width=0.9pt},
  vtx/.style={circle, inner sep=2.2pt}
]

\coordinate (x1) at (0,  8);
\coordinate (x2) at (2, 7);
\coordinate (x3) at (6, 3);
\coordinate (x4) at (8, 0);
\draw[poly] (x1) -- (x2) -- (x3) -- (x4) -- (x1) -- cycle;

\fill[red]              (x1) circle (2.6pt);
\fill[blue]             (x2) circle (2.6pt);
\fill[green!70!black]   (x3) circle (2.6pt);
\fill[orange!90!black]  (x4) circle (2.6pt);

\node[red, right=3pt, font=\footnotesize]  at (x1) {$\bv_{1}(a_1^{1}, \cdot)$};
\node[blue, right=10pt, font=\footnotesize]       at (x2) {$\bv_{1}(a_1^{1}, \cdot)$};
\node[green!70!black,   right=10pt, font=\footnotesize] at (x3) {$\bv_{1}(a_1^{3}, \cdot)$};
\node[orange!90!black,  left=10pt, font=\footnotesize] at (x4) {$\bv_{1}(a_1^{4}, \cdot)$};


\coordinate (x1a) at (0.75,9.5); 
\coordinate (x1b) at (-1, 7); 
\fill[red, cone] (x1) -- (x1a) -- (x1b);
\draw[red, coneborder] (x1a) -- (x1) -- (x1b);
\node[red, above left=1pt, font=\footnotesize]  at (x1) {$\mathcal{R}(a_1^1)$};

\coordinate (x2a) at (3.5,8.5); 
\coordinate (x2b) at (3,9);
\fill[blue, cone] (x2) -- (x2a) -- (x2b);
\draw[blue, coneborder] (x2a) -- (x2) -- (x2b);
\node[blue, xshift=30pt, yshift=30pt, font=\footnotesize] at (x2) {$\mathcal{R}(a_1^2)$};

\coordinate (x3a) at (8.25,4.5); 
\coordinate (x3b) at (8, 5); 
\fill[green!70!black, cone] (x3) -- (x3a) -- (x3b) -- cycle;
\draw[green!70!black, coneborder] (x3a) -- (x3) -- (x3b);
\node[green!70!black, xshift=45pt, yshift=30pt, font=\footnotesize] at (x3) {$\mathcal{R}(a_1^3)$};

\coordinate (x4a) at (9.5,1); 
\coordinate (x4b) at (7, -1);
\fill[orange!90!black, cone] (x4) -- (x4a) -- (x4b) -- cycle;
\draw[orange!90!black, coneborder] (x4a) -- (x4) -- (x4b);
\node[orange!90!black, xshift=10pt, yshift=-10pt, font=\footnotesize] at (x4) {$\mathcal{R}(a_1^4)$};

\end{tikzpicture}
  \caption{The utility polytope constructed as the convex hull of the utility vectors of the games $G(\{A_i\}_{i=1}^n,\{u_i\}_{i=1}^n)$ (left) and $G(\{A_i\}_{i=1}^n,\{v_i\}_{i=1}^n)$ (right) and the best response region for each action.}
  \label{fig:polytopes}
\end{figure}
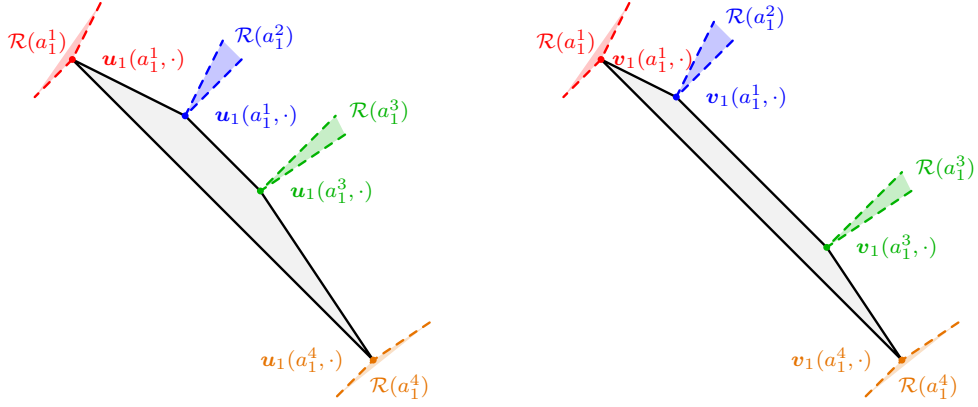

    
    
    
    
    

\subsection{Proof of \cref{lem:Zero_char}}
\label{proof:Zero_char}

\begin{lemma}
Let $G(\{A_i\}_{i=1}^n,\{u_i\}_{i=1}^n)$ and $G'(\{A_i\}_{i=1}^n,\{v_i\}_{i=1}^n)$ be two generic games, and for each player $i$ let $\mathcal P_i=\operatorname{conv}(\{\bu_i(a_i,\cdot)\}_{a_i\in A_i})$ and $\mathcal P'_i=\operatorname{conv}(\{\bv_i(a_i,\cdot)\}_{a_i\in A_i})$. We denote by 
\[
\mathcal N(\mathcal P_i)\mid_{C_i} \coloneqq \{\mathcal N_{\mathcal P_i}(F) \cap C_i\mid F\text{ is a face of }\mathcal P_i\}
\] 
where $\mathcal N_{\mathcal P_i}(F)$ is the normal cone  {of} $F$.  
Then the two games are indistinguishable under best-response feedback if and only if  
\[
\mathcal N(\mathcal P_i)\mid_{C_i}=\mathcal N(\mathcal P'_i)\mid_{C_i} \quad \forall i\in[n]
\]
\end{lemma}
\begin{proof}
Fix a player $i\in[n]$ and let
\[
\mathcal P_i=\operatorname{conv}\bigl(\{u_i(a,\cdot)\}_{a\in A_i}\bigr),
\qquad
\mathcal P'_i=\operatorname{conv}\bigl(\{v_i(a,\cdot)\}_{a\in A_i}\bigr).
\]
For any $\bx\in\Delta(A)$ and any recommended action $a_i\in A_i$, set
\[y := \bx(a_i,\cdot)\in C_i .
\]
For every $y\in C_i$, define
\[
BR_i(y):=\arg\max_{a\in A_i}\langle y,u_i(a,\cdot)\rangle ,
\]
such that  $BR_i(\bx,a_i) = BR_i(y)$. For every $a\in A_i$, we have
\begin{equation}\label{eq:cone-is-BR-region}
N_{\mathcal P_i}\bigl(u_i(a,\cdot)\bigr)\cap C_i
\;=\;
\{\,y\in C_i:\ a\in BR_i(y)\,\},
\qquad \forall a\in A_i,
\end{equation}
and equivalently,
\begin{equation}\label{eq:BR-from-cones}
BR_i(y)
\;=\;
\{\,a\in A_i:\ y\in N_{\mathcal P_i}\bigl(u_i(a,\cdot)\bigr)\cap C_i\,\},
\qquad \forall y\in C_i .
\end{equation}

Therefore, two games are indistinguishable under best-response feedback for player $i$, if and only if they have the same collections of restricted normal cones at extreme points, i.e.
\[
\bigl\{N_{\mathcal P_i}\bigl(u_i(a,\cdot)\bigr)\cap C_i\bigr\}_{a\in A_i}
=
\bigl\{N_{\mathcal P'_i}\bigl(v_i(a,\cdot)\bigr)\cap C_i\bigr\}_{a\in A_i}.
\]
Finally, for any face $F$ of a polytope,
\[
N_{\mathcal P_i}(F)=\bigcap_{z\in\mathrm{ext}(F)} N_{\mathcal P_i}(z),
\]
hence after intersecting with $C_i$,
\[
N_{\mathcal P_i}(F)\cap C_i=\bigcap_{z\in\mathrm{ext}(F)}\bigl(N_{\mathcal P_i}(z)\cap C_i\bigr).
\]
This shows that equality of the restricted normal cones at vertices is equivalent to equality of the
restricted normal fan, i.e.
\[
\mathcal N(\mathcal P_i)\mid_{C_i}=\mathcal N(\mathcal P'_i)\mid_{C_i}.
\]
Since this holds for every $i\in[n]$, the two games are indistinguishable under best-response feedback
if and only if $\mathcal N(\mathcal P_i)\mid_{C_i}=\mathcal N(\mathcal P'_i)\mid_{C_i}$ for all $i\in[n]$.
\end{proof}

\subsection{Proof of \cref{prop:global_after_restriction}}
\label{proof:global_after_restriction}
\begin{proposition}
Let $\mathcal P,\mathcal Q\subset\mathbb R^d$ be two polytopes and let $C\subset\mathbb R^d$ be a closed full-dimensional polyhedral cone
(equivalently, $C^\circ$ is pointed). Suppose that $\mathcal N(\mathcal P)\mid_C=\mathcal N(\mathcal Q)\mid_C$.
Define $\tilde{\mathcal P}:=\mathcal P+C^\circ$ and $\tilde{\mathcal Q}:=\mathcal Q+C^\circ$.
Then $\tilde{\mathcal P}$ and $\tilde{\mathcal Q}$ are normally equivalent polyhedra, i.e., $\mathcal N(\tilde{\mathcal P})=\mathcal N(\tilde{\mathcal Q})$.
Conversely, if $\mathcal N(\tilde{\mathcal P})=\mathcal N(\tilde{\mathcal Q})$, then $\mathcal N(\mathcal P)\mid_C=\mathcal N(\mathcal Q)\mid_C$.
\end{proposition}

\begin{proof}
Let $\tilde{\mathcal P} := \mathcal P + C^\circ$ and $\tilde{\mathcal Q} := \mathcal Q + C^\circ$.
Since $\operatorname{rec}(\tilde{\mathcal P})=\operatorname{rec}(\tilde{\mathcal Q})=C^\circ$, their support functions $h_P(y) := \sup_{x \in P} \langle y, x \rangle$ are finite exactly on $(C^\circ)^\circ = C$, in particular
\[
h_{C^\circ}(y)=
\begin{cases}
0, & y\in C,\\
+\infty, & y\notin C.
\end{cases}
\]
Hence for any $y\in C$,
\[
h_{\tilde{\mathcal P}}(y)=h_{\mathcal P}(y)+h_{C^\circ}(y)=h_{\mathcal P}(y),
\qquad
h_{\tilde{\mathcal Q}}(y)=h_{\mathcal Q}(y)+h_{C^\circ}(y)=h_{\mathcal Q}(y).
\]
Moreover, for every $y\in C$, 
\begin{align}
\label{eq:sum_faces}
\operatorname{face}_{\tilde{\mathcal P}}(y)
=
\operatorname{face}_{\mathcal P}(y)+\operatorname{face}_{C^\circ}(y),
\qquad
\operatorname{face}_{\tilde{\mathcal Q}}(y)
=
\operatorname{face}_{\mathcal Q}(y)+\operatorname{face}_{C^\circ}(y).
\end{align}

We first want to show that :  
\[ 
N_{\tilde{\mathcal P}}\bigl(\operatorname{face}_{\tilde{\mathcal P}}(y)\bigr)=N_{\tilde{\mathcal Q}}\bigl(\operatorname{face}_{\tilde{\mathcal Q}}(y)\bigr) \quad \forall y\in C
\]
To do so, we make the following observations. First, using the standard normal-cone identity for Minkowski sums (applied to $F=\operatorname{face}_{\mathcal P}(y)$ and $G=\operatorname{face}_{C^\circ}(y)$), 
\[
N_{\mathcal P+C^\circ}(F+G)=N_{\mathcal P}(F)\cap N_{C^\circ}(G),
\]
We obtain from \cref{eq:sum_faces}: 
\begin{align}
    \label{eq: factor cone}
    N_{\tilde{\mathcal P}}\bigl(\operatorname{face}_{\tilde{\mathcal P}}(y)\bigr) =
    N_{\mathcal P}\bigl(\operatorname{face}_{\mathcal P}(y)\bigr)
    \cap
    N_{C^\circ}\bigl(\operatorname{face}_{C^\circ}(y)\bigr),\quad \forall y \in C
\end{align}
and similarly for $\tilde{\mathcal Q}$.

Second, by assumption $\mathcal N(\mathcal P)\mid_C=\mathcal N(\mathcal Q)\mid_C$, the cone collections on $C$ coincide.
In particular, for every $y\in C$,
\begin{align}
N_{\mathcal P}\bigl(\operatorname{face}_{\mathcal P}(y)\bigr)\cap C
=
N_{\mathcal Q}\bigl(\operatorname{face}_{\mathcal Q}(y)\bigr)\cap C.
\end{align}

This is not enough since we want to show full equality (without the restriction to $C$). For this we observe that since the normal fan of $C^\circ$ partitions $(C^\circ)^\circ=C$, we have
\begin{align}
N_{C^\circ}\bigl(\operatorname{face}_{C^\circ}(y)\bigr)\subseteq C
\qquad \text{for all }y\in C.
\end{align}
Therefore, intersecting the left hand side of \cref{eq: factor cone} with $C$ does not change the Minkowski-sum identity, and we get 
\begin{align*}
N_{\tilde{\mathcal P}}\bigl(\operatorname{face}_{\tilde{\mathcal P}}(y)\bigr)
&= \left( N_{\mathcal P}\bigl(\operatorname{face}_{\mathcal P}(y)\bigr)\cap C \right)
\cap N_{C^\circ}\bigl(\operatorname{face}_{C^\circ}(y)\bigr) \\
&= \left( N_{\mathcal Q}\bigl(\operatorname{face}_{\mathcal Q}(y)\bigr)\cap C \right)
\cap N_{C^\circ}\bigl(\operatorname{face}_{C^\circ}(y)\bigr) \\
&= N_{\tilde{\mathcal Q}}\bigl(\operatorname{face}_{\tilde{\mathcal Q}}(y)\bigr),
\qquad \forall y\in C.
\end{align*}

Finally, since all normal cones of $\tilde{\mathcal P}$ and $\tilde{\mathcal Q}$ are contained in $C$ , and C is a closed cone, the above equality for all $y\in C$ implies
\[
\mathcal N(\tilde{\mathcal P})=\mathcal N(\tilde{\mathcal Q}),
\]
and thus $\tilde{\mathcal P}$ and $\tilde{\mathcal Q}$ are normally equivalent polyhedra.

\smallskip

\noindent\emph{Conversely, assume $\mathcal N(\tilde{\mathcal P})=\mathcal N(\tilde{\mathcal Q})$.}
Since $C$ is full-dimensional, for any $y\in\operatorname{int}(C)$ we have $\operatorname{face}_{C^\circ}(y)=\{0\}$ and $N_{C^\circ}(\{0\})=C$.
Hence for all $y\in\operatorname{int}(C)$,
\[
N_{\tilde{\mathcal P}}\bigl(\operatorname{face}_{\tilde{\mathcal P}}(y)\bigr)
=
N_{\mathcal P}\bigl(\operatorname{face}_{\mathcal P}(y)\bigr)\cap C,
\qquad
N_{\tilde{\mathcal Q}}\bigl(\operatorname{face}_{\tilde{\mathcal Q}}(y)\bigr)
=
N_{\mathcal Q}\bigl(\operatorname{face}_{\mathcal Q}(y)\bigr)\cap C.
\]
Therefore, $\mathcal N(\tilde{\mathcal P})=\mathcal N(\tilde{\mathcal Q})$ implies that for all $y\in\operatorname{int}(C)$,
\[
N_{\mathcal P}\bigl(\operatorname{face}_{\mathcal P}(y)\bigr)\cap C
=
N_{\mathcal Q}\bigl(\operatorname{face}_{\mathcal Q}(y)\bigr)\cap C.
\]
Since $\mathcal N(\mathcal P)\mid_C$ and $\mathcal N(\mathcal Q)\mid_C$ are polyhedral fans supported on $C$,
equality on $\operatorname{int}(C)$ implies $\mathcal N(\mathcal P)\mid_C=\mathcal N(\mathcal Q)\mid_C$.
Indeed, each maximal cone of a fan supported on $C$ has nonempty intersection with $\operatorname{int}(C)$, and the fan is uniquely determined by its maximal cones.
Thus, equality of the normal-cone maps on $\operatorname{int}(C)$ forces the maximal cones of $\mathcal N(\mathcal P)\mid_C$ and $\mathcal N(\mathcal Q)\mid_C$ to coincide, hence the fans coincide.

Equivalently, one may argue at the level of vertices: for a polytope, every vertex is exposed, and when $C$ is full-dimensional each vertex normal cone intersects $\operatorname{int}(C)$.
Knowing the cones on $\operatorname{int}(C)$ therefore determines the restricted normal cones at vertices, and the restricted normal cones of all faces follow from the identity
\[
N_{\mathcal P}(F)=\bigcap_{v\in\mathrm{Vert}(F)} N_{\mathcal P}(v),
\]
after intersecting with $C$.

\end{proof}

\subsection{Explicit normal-equivalence of polyhedra in $\mathcal H$-representation}
\label{proof:h_rep_eq}

We say that two polyhedra are \emph{normally equivalent} if they share the same normal fan, and we denote by $\mathrm{cl}(P)$ the set of all polyhedra that are normally equivalent to $P$.
From the \cref{thm:ubr_characterization}, we obtained a complete characterization of the invariance set up to normal equivalence. The remaining question is therefore how to explicitly describe the collection of all polyhedra that are normally equivalent to the polarized polyhedron $\tilde P$.

\begin{proposition}
Let  
\[
P=\{y\in\mathbb R^d \mid \langle n_i,y\rangle \le h_i^P \ \forall i\in\llbracket1,M\rrbracket\},
\]
and let  
\[
Q^P(h)=\{y\in\mathbb R^d \mid \langle n_i,y\rangle \le h_i \ \forall i\in\llbracket1,M\rrbracket\}.
\]
Then
\[
\mathrm{cl}(P)=\{Q^P(h)\mid h\in C_P\},
\]
where $C_P$ is a set that we fully characterize in the proof.
\end{proposition}

\begin{proof}
Let us denote by $\mathcal I^*(h^P)\subset\binom{[M]}{d}$ the collection of index sets such that  

\[
\forall I\in\mathcal I^*, \quad N_I=(n_i)_{i\in I} \ \text{is invertible and} \ N_I^{-1}h_I\in P.
\]

The set $\mathcal I^*(h^P)$ defines the vertices of $P$.

Under a genericity assumption, $P$ is simple, and we may rewrite the condition  
$N_I^{-1}h_I\in P$ as

\[
\langle n_i, N_I^{-1}h_I\rangle < h_i \quad \forall i\notin I.
\]

We now use the following classical result .

\begin{lemma}
Let $Q$ be normally equivalent to $P$.  
Then, up to a reordering of indices, $Q$ admits an $\mathcal H$-representation of the form
\[
Q=\{y\in\mathbb R^d:\langle n_i,y\rangle\le h_i,\ i\in\llbracket1,M\rrbracket\}
\]
for some $h\in\mathbb R^M$.
\end{lemma}

The statement remains valid when $P$ is not full-dimensional by working in the affine hull of $P$.

From the previous lemma, all polyhedra normally equivalent to $P$ are of the form $Q(h)$.  
Since $Q$ and $P$ have the same normal fan, they have the same maximal cones, and therefore

\[
\{\operatorname{cone}(n_i)_{i\in I}\mid I\in\mathcal I^*(h^P)\}
=
\{\operatorname{cone}(n_i)_{i\in I}\mid I\in\mathcal I^*(h^Q)\}.
\]

This implies, up to a reordering of indices, $\mathcal I^*(h^Q)=\mathcal I^*(h^P).$ Hence $h^Q$ must satisfy

\[
\forall I\in\mathcal I^*(h^P), \quad 
\langle n_i, N_I^{-1}h_I\rangle < h_i \quad \forall i\notin I,
\]

and

\[
\forall I\in\binom{[M]}{d}\setminus\mathcal I^(h^P), \ \text{such that } N_I \text{ is invertible}, \ 
\exists i\in\llbracket1,M\rrbracket \text{ such that }
\langle n_i, N_I^{-1}h_I\rangle > h_i.
\]

\end{proof}

\subsection{Extreme points polarization}
\label{proof:ext_points_polarization}

\begin{lemma}
\label{lem:ext_points_polarization}
Let $P\subset\mathbb R^d$ be polytope and let $C\subset\mathbb R^d$ be a closed full dimensional cone. 
Define $\tilde P:=P+C^\circ$.

\smallskip
(i) One always has $\operatorname{ext}(\tilde P)\subseteq \operatorname{ext}(P)$.

\smallskip
(ii) Moreover, if every vertex of $P$ is exposed by some direction in $\operatorname{int}(C)$, i.e.,
\[
\forall v\in\operatorname{ext}(P)\ \exists y\in \operatorname{int}(C) \text{ such that } \operatorname{face}_P(y)=\{v\},
\]
then $\operatorname{ext}(\tilde P)=\operatorname{ext}(P)$.

\end{lemma}

\begin{proof}
(i) Let $v\in\operatorname{ext}(\tilde P)$. Write $v=p+k$ with $p\in P$ and $k\in C^\circ$.  
If $k\neq 0$, then $p\in\tilde P$ (since $0\in C^\circ$) and $p+2k\in\tilde P$ (since $2k\in C^\circ$) are different, and
\[
v=p+k=\tfrac12\,p+\tfrac12\,(p+2k),
\]
which contradicts that $v$ is extreme in $\tilde P$. Hence $k=0$ and $v=p\in P$.
If $p\notin\operatorname{ext}(P)$, then $p$ is a nontrivial convex combination of two distinct points in $P\subseteq \tilde P$, again contradicting that $v$ is extreme in $\tilde P$. 
Therefore $v\in\operatorname{ext}(P)$, proving $\operatorname{ext}(\tilde P)\subseteq \operatorname{ext}(P)$.

\smallskip
\smallskip
(ii) Let $v\in\operatorname{ext}(P)$ and pick $y\in \operatorname{int}(C)$ such that $\operatorname{face}_P(y)=\{v\}$.
Since $y\in \operatorname{int}(C)$, we have $\langle y,k\rangle<0$ for all $k\in C^\circ\setminus\{0\}$, hence
\[
\operatorname{face}_{C^\circ}(y)=\{0\}\qquad\text{and}\qquad h_{C^\circ}(y)=0.
\]
Therefore,
\[
\operatorname{face}_{\tilde P}(y)
=
\operatorname{face}_{P+C^\circ}(y)
=
\operatorname{face}_{P}(y)+\operatorname{face}_{C^\circ}(y)
=
\{v\},
\]
so $v\in\operatorname{ext}(\tilde P)$. Thus $\operatorname{ext}(P)\subseteq \operatorname{ext}(\tilde P)$, and together with (i) we get
$\operatorname{ext}(\tilde P)=\operatorname{ext}(P)$.

\end{proof}

\section{Omitted proofs from \cref{sec:learning_alg}}
\label{appx:omitted_proofs_learning_alg}
\subsection{Proof of \cref{lem:sign pattern}}
\label{proof:sign pattern}
\begin{lemma}\label{lem:sign pattern}
\Cref{alg:sign_pattern} learns the sign pattern sets $\mathcal{P}_{i}(a_i, a_i')$ and $\mathcal{N}_{i}(a_i, a_i')$ for a fixed agent $i$ and every $a_i \neq a_i \in A_i$ with $|A_{-i}|$ recommendations $\bx$.  
\end{lemma}

\begin{proof}
    We first argue that is possible to construct $\bx^{(k)}$. Fix an order of the elements of the set $A_{-i} = [a_{-i}^1, a_{-i}^2, \dots]$. Construct $\tilde{\bx}^{(k)}$ such that $\tilde{\bx}(a_i, a_{-i}^k) = 1$ and $\tilde{\bx}(a_i, a_{-i}^{k'}) = 0$ for all $k' \neq k$. Then we normalize to obtain or distribution $\bx^{(k)} = \tilde{\bx}^{(k)}/ \|\tilde{\bx}^{(k)}\|_1$.
    By definition of the quantal-response feedback, 
    \[a_i' \in QR_i(\bx^{(k)},a_i)\iff \langle \bx^{(k)}(a_i,\cdot),\bw_i(a_i,a_i')\rangle\ge0.
    \]
    Since $\bx^{(k)} (a_i, \cdot) = e_k/m$, we have $\langle \bx^{(k)}(a_i,\cdot),\bw_i(a_i,a_i')\rangle = \frac{1}{m}(\bw_i(a_i,a_i'))_k$,
hence $a_i' \in QR_i(\bx^{(k)},a_i)\iff(\bw_i(a_i,a_i'))_k\ge0$, which proves the claim.
\end{proof}

\subsection{Proof of \cref{lem:sampling_complexity}}
\label{proof:sampling_complexity}
\begin{lemma}
Consider the game $G(\{A_i\}_{i=1}^n,\{u_i\}_{i=1}^n)$ and assume that all utility differences are uniformly bounded away from zero and infinity, 
\(
c \;\le\; \big|u_i(a_i',\ba_{-i}) - u_i(a_i,\ba_{-i})\big| \;\le\; C 
\)
for some constants $0 < c < C$, for all agents $i \in [n]$, action pairs $a_i \neq a_i' \in A_i$, and profiles $\ba_{-i} \in A_{-i}$. Fix an agent $i$ and a pair of actions $a_i \neq a_i'$. Then for each recommendation mechanism $\bx$ constructed as in \cref{alg:binary_ratio} the moderator needs at most 
\[m_i e^{\beta C} \ln(1/\delta)\]
samples to verify if $a_i' \in QR_i(\bx, a_i)$ with probability at least $1-\delta$. 
\end{lemma}

\begin{proof}
    First note that the distributions $\bx$ constructed in \cref{alg:binary_ratio} assign probability $1$ to samples $\ba \sim \bx$ with the action $a_i$. So the sample complexity comes entirely from the quantal-response choice model of the agent in \cref{eq:quantal_prob}. 

    By the construction of $\bx$ in \cref{alg:binary_ratio}, for positive and negative indices $(p, j)$ we have: 
    \begin{align}
        \varphi_{i}(a_i, a_i', \bx) = \frac{(\bw_{i}(a_i, a_i'))_p + \tau (\bw_{i}(a_i, a_i'))_j}{1+\tau} \leq C
    \end{align}

   Let $Y_1, \dots, Y_T$ be iid samples from $P_{\bx}(\cdot |a_i)$ (cf. \cref{eq:quantal_prob}). If $a_i' \in QR_{i}(\bx, a_i)$, then 
   \begin{align}
       P_{\bx}(a_i'|a_i) = \frac{e^{\beta\varphi(a_i, a_i', \bx)}}{\sum_{b \in QR(\bx, a_i)} \exp(\beta\varphi(a_i, b, \bx))} 
       \geq 
       \frac{\exp(\beta\varphi(a_i, a_i', \bx))}{m_i \exp(\beta C)} = \frac{1}{m_i}\exp(-\beta(C - \varphi(a_i, a_i', \bx))
   \end{align}
   where the first inequality follows by $k \leq m_i$ and $\varphi_{i}(a_i, b, \bx) \leq C$ for all $b \in A_i$. 

   Then, for every $\delta \in (0, 1)$ and number of samples 
   \begin{align}
       T \geq k \exp(\beta C) \ln(1/\delta) 
   \end{align}
   The probability that we see the deviation $a_i'$ in $T$ samples: 
   \begin{align}
       \operatorname{Pr}(\exists t\leq T: Y_t = a_i') \geq 1-\delta.
   \end{align}

   Conversely, we can verify that $a_i' \notin QR(\bx, a_i)$ with probability $1 - \delta$ if we do not see it after $T$ samples.
\end{proof}

\end{document}